\documentclass[twocolumn,twoside]{IEEEtran}
\usepackage{xcolor,soul,framed} 

\colorlet{shadecolor}{yellow}
\usepackage[pdftex]{graphicx}
\graphicspath{{../pdf/}{../jpeg/}}
\DeclareGraphicsExtensions{.pdf,.jpeg,.png}

\usepackage[cmex10]{amsmath}
\usepackage{booktabs}
\usepackage{multirow}
\usepackage{array}
\usepackage{mdwmath}
\usepackage{mdwtab}
\usepackage{eqparbox}
\usepackage{url}
\usepackage{cite}
\usepackage{amsfonts}
\usepackage{amsmath, amsthm} 
\newtheorem{theorem}{Theorem} 
\newtheorem{corollary}{Corollary}
\newtheorem{remark}{Remark}
\usepackage{mathrsfs}
\usepackage{tabularx}
\usepackage[ruled]{algorithm2e}
\usepackage{hyperref}
\usepackage[normalem]{ulem}
\usepackage{amssymb}
\usepackage{subfigure}

\newcommand{\bs}[1]{\ensuremath{{\boldsymbol{#1}}}}

\IEEEoverridecommandlockouts
\begin{document}
	\title{DRCC-LPVMPC: Robust Data-Driven Control for Autonomous Driving and Obstacle Avoidance~\thanks{The work was supported in part by the National Science Foundation (NSF) under Award CMMI 2146056 (K. Yu). ({\em Corresponding author: Kaiyan Yu}).}}
	\author{Shiming Fang\thanks{S. Fang and K. Yu are with the Department of Mechanical Engineering, Binghamton University, Binghamton, NY 13902 USA (email: {sfang10@binghamton.edu}; {kyu@binghamton.edu}). X. Li is with the Department of Mechanical and Aerospace Engineering, Syracuse University, Syracuse, NY 13244  (e-mail: xli428@syr.edu). C. Wu is with the National Center for Applied Mathematics, Chongqing Normal University, Chongqing, 401331, China (e-mail: changzhiwu@gzhu.edu.cn).}, Xilin Li\thanks{Source codes are available at \href{https://github.com/Binghamton-ACSR-Lab/DRCCLPVMPC}{https://github.com/Binghamton-ACSR-Lab/DRCCLPVMPC}.}, Changzhi Wu, and Kaiyan Yu,~\IEEEmembership{Member,~IEEE}}
\maketitle

\begin{abstract}
	Safety in obstacle avoidance is critical for autonomous driving. While model predictive control (MPC) is widely used, simplified prediction models such as linearized or single-track vehicle models introduce discrepancies between predicted and actual behavior that can compromise safety. This paper proposes a distributionally robust chance-constrained linear parameter-varying MPC (DRCC-LPVMPC) framework that explicitly accounts for such discrepancies. The single-track vehicle dynamics are represented in a quasi-linear parameter-varying (quasi-LPV) form, with model mismatches treated as additive uncertainties of unknown distribution. By constructing chance constraints from finite sampled data and employing a Wasserstein ambiguity set, the proposed method avoids restrictive assumptions on boundedness or Gaussian distributions. The resulting DRCC problem is reformulated as tractable convex constraints and solved in real time using a quadratic programming solver. Recursive feasibility of the approach is formally established. Simulation and real-world experiments demonstrate that DRCC-LPVMPC maintains safer obstacle clearance and more reliable tracking than conventional nonlinear MPC and LPVMPC controllers under significant uncertainties.
\end{abstract}


%
\IEEEpeerreviewmaketitle

\section{Introduction}

\IEEEPARstart{S}{afety} is a critical challenge in autonomous driving, particularly in dynamic scenarios such as collision avoidance, which involve complex vehicle models and unpredictable disturbances. Model predictive control (MPC) has become a popular framework because it explicitly handles system and input constraints while benefiting from advances in machine learning and computation~\cite{zhang2020near,hewing2019cautious}. Its use in autonomous driving has been widely studied~\cite{kensbock2023scenario,williams2018information}, motivated by the growing demand for AI-based safe control architectures.  

Dynamic vehicle models capture lateral dynamics better than simple kinematic models, but discrepancies between even detailed four-wheel models and real vehicle behavior remain problematic in high-speed driving. Single-track models (STMs), which approximate four-wheel dynamics as a bicycle model, strike a balance between fidelity and computational cost. Nonlinear MPC (NMPC) is often applied for such models, but its high computational burden and feasibility issues in real-time obstacle avoidance remain challenging~\cite{nezami2024obstacle}.  

In MPC-based obstacle avoidance, obstacles are often simplified into geometric shapes to keep constraints tractable. However, safety requires more than collision avoidance alone: a vehicle should begin adjusting its trajectory smoothly once an obstacle is detected, instead of executing abrupt turns or hard braking when the obstacle is already near. To enable smoother maneuvers, bounding-box methods in the Tangent-Normal-Binormal (TNB) frame have been explored~\cite{alcala2020lpv}.  

To address these challenges, a common approach is to use re-linearization methods to simplify nonlinear dynamics. Accordingly, this paper adopts a quasi-linear parameter-varying (quasi-LPV) approach~\cite{bokor2005linear} to linearize the dynamics and constraints of the nonlinear vehicle model. LPV modeling represents nonlinear systems in a linear time-varying framework by approximating the nonlinear model with adaptive linear representations. In quasi-LPV, the scheduling parameter vector is defined as a function of selected system states and inputs, ensuring system matrices are updated dynamically. By converting NMPC into a linear formulation, LPV-based MPC (LPVMPC) reduces complexity to solving a quadratic program (QP), while still adapting the model matrices online.

LPVMPC has demonstrated strong performance in tasks such as racing~\cite{alcala2020autonomous} and path following~\cite{tian2022gain}, yet challenges persist in predicting scheduling parameters and ensuring robustness. quasi-LPV-MPC methods have also been extended to UAVs: \cite{rodriguez2023qlpv} applied a 12-state model for quadrotor tracking, and \cite{samir2024velocity} incorporated obstacle avoidance and stability constraints for fixed-wing UAVs. These studies highlight the versatility of quasi-LPV modeling but remain limited to simulation, without addressing uncertainty or real-time feasibility.

While iterative quasi-LPV-MPC schemes can converge to nonlinear dynamics with sufficient iterations~\cite{hespe2021convergence}, real-time implementation is impractical under fast sampling or complex tasks. Even with convergence, the bicycle model remains an approximation of the unknown full dynamics. To enable real-time control without iterative refinement, we adopt a pragmatic strategy: scheduling parameters at each step are taken from the previous optimal solution~\cite{nezami2024obstacle}, reducing computation time but introducing model mismatch and sensor noise. These uncertainties can compromise safety by producing trajectories close to obstacle boundaries.


The concern of disturbances further motivates robust MPC (RMPC). For instance, \cite{chen2025tube} proposed Tube-RMPC for road-tire adhesion uncertainty, and \cite{liu2021matrix} designed robust feedback laws via convex constraints. However, these methods often yield conservative solutions due to overestimated uncertainty bounds. To reduce conservatism, distributionally robust optimization (DRO) has emerged as a promising alternative. Recent work has demonstrated DRO’s capability to manage uncertainties in motion planning and control~\cite{10465640}, as well as robotics~\cite{10610404}, though nonlinear and nonconvex formulations still hinder real-time applications.

Building on DRO, recent studies have extended robust optimization to distributionally robust chance-constrained (DRCC) problems in autonomous robotics. For instance,~\cite{10610404} ensured safe navigation in crowded environments via risk metrics and DRCC-based safe-corridor constraints to handle perception uncertainty, while~\cite{hakobyan2020learning} integrated Gaussian Process regression with DRO for robust motion control under unknown obstacle dynamics. In~\cite{chen2021constrained}, chance-constrained Gaussian belief space planning enabled real-time, uncertainty-aware motion planning. More recent efforts have combined DRCC with MPC, such as data-driven DRCC-MPC for adaptive cruise control~\cite{drcc_mpc_cacc_2024}, DRCC-based safe-corridor trajectory optimization~\cite{drcc_traj_opt_2024}, and Wasserstein DRCC-MPC for collision avoidance~\cite{drcc_wasserstein_2024}. However, most retain nonlinear models and nonconvex constraints, leading to high computational cost and limited real-time applicability.

A primary challenge in solving DRCC problems is therefore tractability. Data-driven approaches provide efficient DRCC solutions when uncertainty distributions are unknown, often using Wasserstein ambiguity sets~\cite{chen2022wasserstein}. For example, \cite{liu2025data} integrates Wasserstein-ball ambiguity sets into power-flow chance constraints, \cite{chen2025bayesian} leverages Bayesian intervals for finite-sample guarantees, \cite{comden2021moment} employs moment-based sets for efficiency, and \cite{bahari2025data} applies DRO to multi-agent systems under chance constraints.

While LPVMPC has been widely applied to trajectory tracking, most schemes assume known disturbance models and are not robust to distributional ambiguity~\cite{kumar2025realtime}. Conversely, DRCC provides safety guarantees under distributional uncertainty but has rarely been integrated into LPV-based vehicle control. To the best of our knowledge, this work is the first to combine data-driven DRCC with LPVMPC, enabling robust control under probabilistic safety constraints learned from sample data and filling an important gap in autonomous driving research.

This paper introduces a novel framework, DRCC-LPVMPC, for autonomous vehicle control under uncertainty. The method accounts for discrepancies between simplified linearized models and real vehicle dynamics, as well as additive disturbances from sensors and localization. We adopt a data-driven strategy: discrepancies between measured sensor states and quasi-LPV predictions are used to generate empirical disturbance samples, forming ambiguity sets for the DRCC formulation. These ambiguity sets constrain the Wasserstein distance between empirical and true distributions, yielding an infinite-dimensional DRCC problem. This can be reformulated as finite-dimensional convex constraints resembling conditional-value-at-risk (CVaR) constraints with right-hand uncertainty~\cite{xie2021distributionally}, solvable efficiently with a QP solver. The full framework is outlined in Fig.~\ref{fig:Overview}.  
Numerical and experimental results demonstrate that DRCC-LPVMPC improves robustness and safety in obstacle avoidance compared with NMPC and standard LPVMPC, especially in multi-obstacle scenarios with significant uncertainties.

\begin{figure}[t!]
	\begin{center}
		\includegraphics[width=0.97\linewidth]{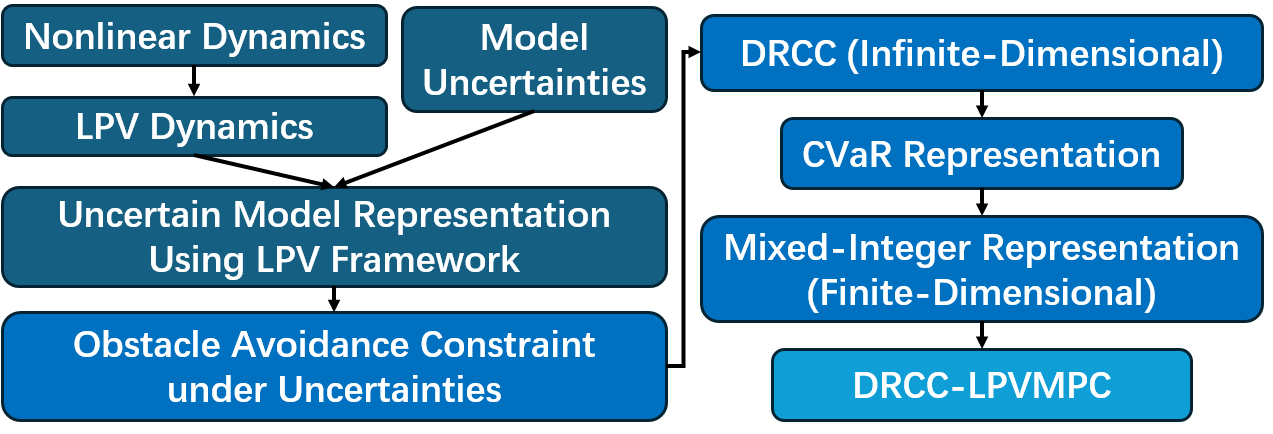}
        \vspace{-3mm}
		\caption{Overview of the derivation and design process of DRCC-LPVMPC, highlighting the transition from the nonlinear vehicle model to the quasi-LPV representation with additive uncertainties, and the formulation of the chance-constrained problem using the right-hand CVaR constraint.}\label{fig:Overview}
	\end{center}
\end{figure}

The main contributions of this paper are 1) the development of DRCC-LPVMPC, a data-driven LPVMPC framework that enhances safety and robustness in autonomous driving by explicitly accounting for model uncertainties and additive disturbances, achieving real-time performance via QP-based optimization. 2) Vehicle dynamics are represented in a quasi-LPV form derived from a simplified STM, with model discrepancy and sensor noise captured as sampled additive uncertainties, enabling efficient real-time control without full nonlinear models. 3) The method’s recursive feasibility is formally established, and simulations demonstrate larger obstacle clearances and more reliable tracking than NMPC and LPVMPC under significant uncertainties, without assuming bounded or Gaussian distributions. 4) Real-world experiments further validate the approach, demonstrating robust obstacle avoidance where LPVMPC cannot avoid obstacles.

This paper is organized as follows: Section~\ref{sec:Dyanmics} presents the quasi-LPV vehicle dynamics model. Section~\ref{sec:ObstacleAvoidance} introduces the linearized obstacle avoidance and LPVMPC formulation. Section~\ref{sec:DRCC-LPVMPC} develops the finite-dimensional DRCC constraints and the DRCC-LPVMPC formulation with recursive feasibility. Section~\ref{sec:Setup} details the nonlinear dynamics and the experimental setup for LPVMPC, DRCC-LPVMPC, and NMPC. Section~\ref{sec:result} compares NMPC, LPVMPC, and DRCC-LPVMPC, highlighting DRCC-LPVMPC’s enhanced safety and feasibility. Section~\ref{sec:RealExp} presents real-world results on a 1:5 scaled vehicle, demonstrating robust obstacle avoidance. Section~\ref{sec:concl} concludes with key findings and future directions.

\section{Nonlinear Dynamics and LPV Modeling}
\label{sec:Dyanmics}

\subsection{Notations} 

The symbol $\bs{\mathcal{I}}_n$ denotes an $n$ by $n$ identity matrix. The symbol \(\mathbb{Z}_{\geq 0}\) represents the set of nonnegative integers and $\mathbb{Z}_{>0}$ represents the positive integer set. 
The symbol \(\mathbb{R}^{i\times j}\) represents a real matrix with $i$th rows and $j$th columns. Given an integer $n$, we let $[n] := \{1,2,\dots,n\}$. The symbol $\bs{e}_k^n \in \mathbb{R}^{1\times n}$ represents an elementary vector where only the $k$th element is equal to one. We let $\tilde{\bs{w}} := (\tilde{w}^1,\dots,\tilde{w}^N)$ denote the set distribution of the unknown disturbance and $\hat{\bs{w}}_j := (\hat{w}_j^1,\dots,\hat{w}_j^N),j\in[J]$ denote the $j$th set of sampled data. We further denote $\tilde{\bs{e}}^k_w := [\bs{0}_{1\times6},\tilde{\bs{w}}^k]^\top,k\in[N]$ and $\hat{\bs{e}}_{wj}^k := [\bs{0}_{1\times6},\hat{\bs{w}}_j^k]^\top,j\in[J],k\in[N]$ as the elementary disturbance vectors. $\|\cdot\|$ represents the norm operator. 

We denote $R_r$ as the reference trajectory and define a TNB frame based on it. In this frame, the coordinate along the tangential direction of $R_r$ is denoted as $T$, while the coordinate perpendicular to $T$ is represented as $\Lambda$. $S$ is defined as the arc-length traveled along $R_r$ from the starting point to the current position $[T,\Lambda]$. The tangential angle at a given point along $R_r$ is represented by $\theta$. Transformation functions $\Gamma_\text{r}$, $\Gamma_\text{in}^{T}$, $\Gamma_\text{in}^{S}$, and $\Gamma_\text{in}^{\theta}$ map points in the TNB frame to Cartesian coordinates, $T$, $S$, and $\theta$, respectively. Detailed discussions of these notations and their roles appear in Sections~\ref{section:linearizeObsAvoidCons} and~\ref{section:LPVMPC}, with additional symbols explained as introduced.

\subsection{Nonlinear Vehicle Model}

In this paper, we adopt STM within the quasi-LPV framework, as illustrated in Fig.~2, and account for model discrepancies through an unknown-distributed uncertainty term. This approach enables computationally efficient predictive control while addressing uncertainties in both the linearized dynamics and additive disturbances, ensuring robust obstacle avoidance in autonomous driving. 

\begin{figure}[bt!]
	\begin{center}
\includegraphics[width=\linewidth]{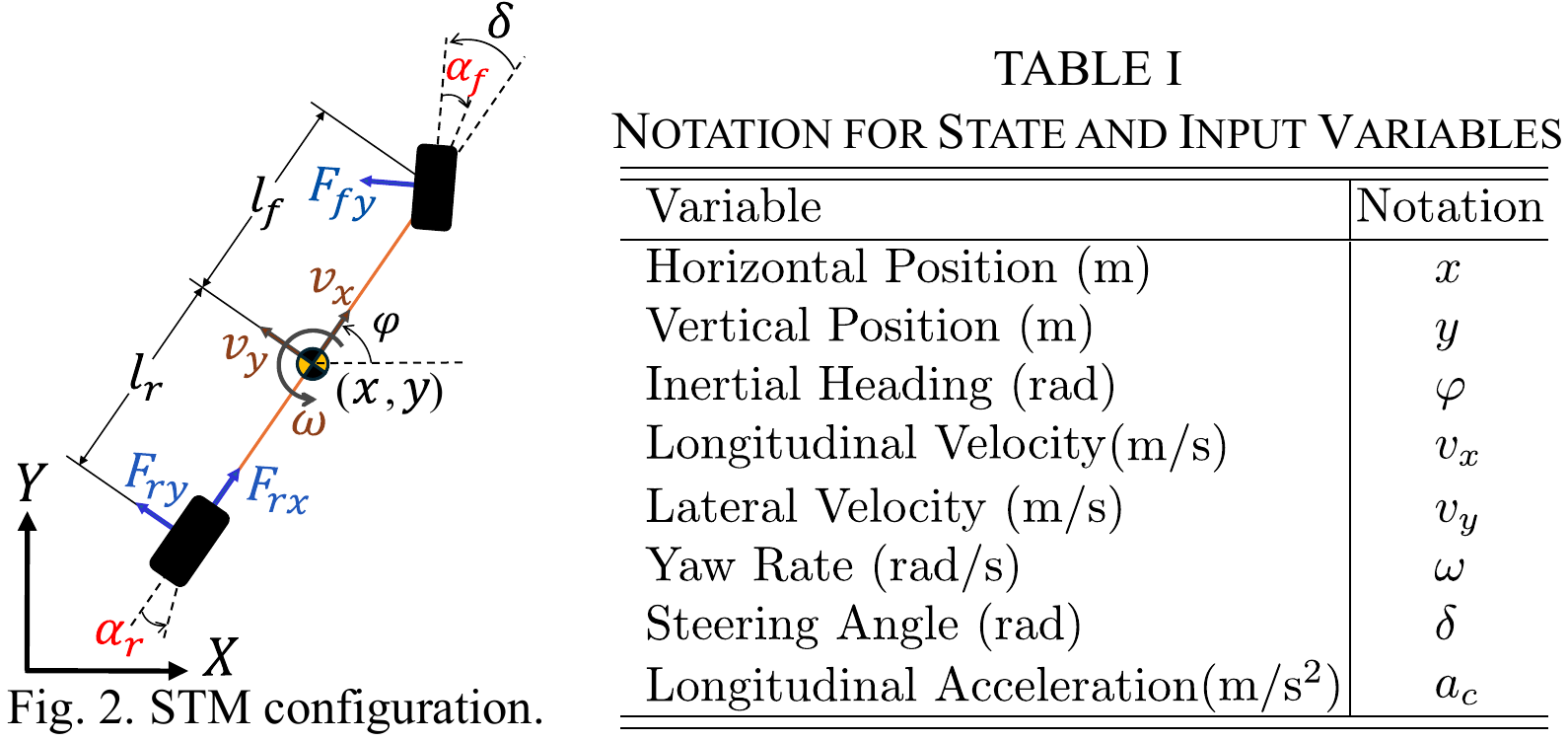}
	\end{center}
\end{figure}

Following~\cite{rajamani2011vehicle}, the vehicle state vector $\bs{z}$ at time step $k$ is defined as \(\bs{z}_{k} = [x_{k},y_{k},\varphi_{k},v_{x_{k}},v_{y_{k}},\omega_{k}]\), representing the position, orientation, linear velocities, and yaw rate of the vehicle's center of gravity. The control input vector $\bs{u}$ consists of the steering angle and acceleration, given by \(\bs{u}_k = [\delta_{k},a_{c_{k}}]\), as summarized in Table I. 
The slip angle at the front tire \(\alpha_f\) and the rear tire \(\alpha_r\) are calculated as:
\begin{equation}\label{eq:alpha}
        \alpha_f = \delta - \frac{l_f\omega + v_y}{v_x},\quad \alpha_r = \frac{l_r \omega - v_y}{v_x}.
\end{equation}
The front and rear lateral forces \(F_{fy}\) and \(F_{ry}\) are calculated based on the cornering stiffness \(C_{\alpha f}\) and \(C_{\alpha r}\): \(F_{fy} = C_{\alpha f} \alpha_f\) and \(F_{ry} = C_{\alpha r} \alpha_r\). The longitudinal force acting on the rear wheel, \(F_{rx}\), is given by: \(F_{rx} = ma_c\), where $m$ is the vehicle mass. Thus, the continuous differential equations governing the vehicle dynamics, represented as \(\dot{\bs{z}}_k=f(\bs{z}_k,\bs{u}_k)\), are formulated as follows:   
\begin{equation}
	\begin{aligned}\label{eq:differential_equation}
        \dot{x}_k &= v_{x_k} \cos \varphi_k - v_{y_k} \sin \varphi_k, \\
        \dot{y}_k &= v_{x_k} \sin \varphi_k + v_{y_k} \cos \varphi_k, \\
        \dot{\varphi}_k &= \omega_k, \\
        \dot{v}_{x_{k}} &= \frac{1}{m} (F_{rx_k} - F_{fy_k} \sin \delta_k + m v_{y_k} \omega_k), \\
        \dot{v}_{y_{k}} &= \frac{1}{m} (F_{ry_k} + F_{fy_k} \cos \delta_k - m v_{x_k} \omega_k), \\
        \dot{\omega}_{k} &= \frac{1}{I_z} (F_{fy_k} l_f \cos \delta_k - F_{ry_k} l_r),
    \end{aligned}
\end{equation}
where $l_f$ and $l_r$ denote the distances from the center of gravity to the front and rear axles, respectively, and $I_z$ is the vehicle’s yaw moment of inertia.

Before presenting the LPV reformulation, we note that the system is considered in a discrete-time setting, with sampling time $\Delta t$, such that the nonlinear dynamics evolve as $\bs{z}_{{k+1}} = \bs{z}_{k} + f(\bs{z}_{k},\bs{u}_{k})\Delta t, k\in \mathbb{Z}_{\geq 0}$.

\subsection{Quasi-LPV Model}
To formulate the quasi-LPV model used in LPVMPC based on Eq.~\eqref{eq:differential_equation}, we first define the scheduling parameter vector as \(\bs{p}_k = [v_{x_k}, v_{y_k}, \delta_k,\varphi_k]\). Using this scheduling parameter, the quasi-LPV representation of Eq.~\eqref{eq:differential_equation} is expressed as: 
\begin{equation}\label{eq:LPV_state_differential}
    \dot{\bs{z}}_{k} = \bs{\hat{A}}(\bs{p}_{k})\bs{z}_{k} + \bs{\bs{\hat{B}}}(\bs{p}_{k})\bs{u}_{k}, \quad k \in \mathbb{Z}_{\geq 0},
\end{equation}
where \(\bs{\hat{A}} \in \mathbb{R}^{6 \times 6}\) and \(\bs{\hat{B}} \in \mathbb{R}^{6 \times 2}\). The state transition matrix \(\bs{\hat{A}}\) is given by: 
\[
\bs{\hat{A}}(\bs{p}_k) :=
\begin{bmatrix}
0 & 0 & 0 & \cos{\varphi_{k}} & -\sin{\varphi_{k}} & 0 \\
0 & 0 & 0 & \sin{\varphi_{k}} & \cos{\varphi_{k}} & 0 \\
0 & 0 & 0 & 0 & 0 & 1 \\
0 & 0 & 0 & 0 & \hat{a}_{45_{k}} & \hat{a}_{46_{k}} \\
0 & 0 & 0 & 0 & \hat{a}_{55_{k}} & \hat{a}_{56_{k}} \\
0 & 0 & 0 & 0 & \hat{a}_{65_{k}} & \hat{a}_{66_{k}}
\end{bmatrix}_{6\times6}
\]
To simplify the representations, we denote:
\begin{equation*}
    \beta_f := \frac{C_{\alpha f}}{m},\quad \beta_r := \frac{C_{\alpha r}}{m},\quad \gamma_f := \frac{C_{\alpha f}l_f}{I_z},\quad \gamma_r := \frac{C_{\alpha r}l_r}{I_z},
\end{equation*}
and the elements in $\bs{\hat{A}}$ are given as:
\begin{equation*}
	\begin{aligned}
        \hat{a}_{45_{k}} &:= \beta_f\frac{\sin{\delta_{k}}}{v_{x_{k}}}, \quad \hat{a}_{46_{k}} := \beta_f\frac{l_f\sin{\delta_{k}}}{v_{x_{k}}} + v_{y_{k}},\\
        \hat{a}_{55_{k}} &:= -\beta_r\frac{1}{v_{x_{k}}} - \beta_f\frac{\cos{\delta_{k}}}{v_{x_{k}}},\\
        \hat{a}_{56_{k}} &:= \beta_r\frac{l_r}{v_{x_{k}}} - \beta_f\frac{l_f\cos{\delta_k}}{v_{x_{k}}} - v_{x_{k}}, \\
        \hat{a}_{65_{k}} &:= \gamma_r\frac{1}{v_{x_{k}}} - \gamma_f\frac{\cos{\delta_k}}{v_{x_{k}}},\quad \hat{a}_{66_{k}} := -\gamma_f\frac{l_f\cos{\delta_{k}}}{v_{x_{k}}} - \gamma_r\frac{l_r}{v_{x_{k}}}.
    \end{aligned}
\end{equation*}
The input matrix $\bs{\hat{B}}$ is:
\[
\bs{\hat{B}}(\bs{p}_k) :=
\begin{bmatrix}
0 & 0 & 0 & \hat{b}_{41_k} & \hat{b}_{51_k} & \hat{b}_{61_k} \\
0 & 0 & 0 & 1 & 0 & 0 \\
\end{bmatrix}^\top_{6\times2}
\]
and the elements in $\bs{\hat{B}}$ are $
    \hat{b}_{41_k} := -\beta_f\sin{\delta_k}$, $\hat{b}_{51_k} := \beta_f\cos{\delta_k}$, and $\hat{b}_{61_k} := \gamma_f\cos{\delta_k}$.
Based on~\cite{bokor2005linear}, this formulation ensures that the system matrices $\bs{\hat{A}}$ and $\bs{\hat{B}}$ depend only on the scheduling parameter vector $\bs{p}_k$.  
The next state $\bs{z}_{k+1}$ can be obtained using the Euler method as follows:
\begin{equation}\label{eq:lpv_rewrite}
	\begin{aligned}
        \bs{z}_{k+1} &= \bs{z}_k + \dot{\bs{z}}_k\Delta t\\
        &= \bs{z}_k + (\bs{\hat{A}}(\bs{p}_k)\bs{z}_k + \bs{\hat{B}}(\bs{p}_k)\bs{u}_k)\Delta t \\
        &= (\bs{\mathcal{I}}_6 + \bs{\hat{A}}(\bs{p}_k)\Delta t)\bs{z}_k + \bs{\hat{B}}(\bs{p}_k)\Delta t\bs{u}_k \\
        &= \bs{A}(\bs{p}_k)\bs{z}_k + \bs{B}(\bs{p}_k)\bs{u}_k,
   	\end{aligned}
\end{equation}
where $\bs{A}(\bs{p}_k) = \bs{\mathcal{I}}_6 + \bs{\hat{A}}(\bs{p}_k)\Delta t$ and $\bs{B}(\bs{p}_k) = \bs{\hat{B}}(\bs{p}_k)\Delta t$ are the discrete-time quasi-LPV system matrices. For simplicity, we denote these matrices as $\bs{A}_k$ and $\bs{B}_k$ for the remainder of this work.

\subsection{Quasi-LPV Model with Uncertainties}
\label{section:nonlinearModelsamples}
Based on~\cite{pacejka2005tire}, the formulations in Eq.~\eqref{eq:alpha} and the computations for $F_{fy}$ and $F_{ry}$ are intentionally simplified to facilitate subsequent linearization steps. While these simplifications yield a reasonable approximation, they inevitably introduce some inaccuracies into the model.
Even a complete four-wheel vehicle model cannot fully capture the complexities of real-world vehicle behavior, and employing an overly detailed model dramatically increases computational costs, thus rendering real-time control impractical. Although the quasi-LPV method allows for efficient linearization and faster solutions via quadratic programming (outperforming NMPC in computational speed), it remains insufficient for addressing the combined challenges posed by complex vehicle dynamics and measurement noise in physical experiments. Consequently, we adopt a STM that strikes a balanced trade-off between accuracy and computational efficiency, making it more suitable for real-time autonomous driving applications.

Additionally, the scheduling parameter vector required for the quasi-LPV formulation is not directly accessible because it contains the steering input $\delta_k$, which is the decision variable to be computed by the MPC optimization. Therefore, the scheduling parameter cannot be evaluated before solving the optimization problem. To address this issue, we adopt a common quasi-LPV MPC implementation strategy and approximate the scheduling parameter using the previously optimized solution, i.e., $\bs{p}_{i|k} \approx \bs{p}^\star_{i-1|k+1}$
where $\bs{p}_{i|k}$ is the predicted scheduling parameter vector for the $i$th iteration at time step $k$, and $\bs{p}^\star_{i-1|k+1}$ is the optimal scheduling parameter vector from the $(i-1)$th iteration at time step $k+1$. This approximation has demonstrated promising results in scenarios where estimation errors are minimal~\cite{alcala2020lpv}.

However, in obstacle avoidance scenarios, these estimation errors can become critical, as inaccuracies are amplified. When combined with additive measurement errors, such uncertainties can cause the actual vehicle path, determined by the predicted control inputs, to deviate, potentially intersecting with obstacles and leading to collisions.

As discussed previously, both the quasi-LPV estimated STM and more complex four-wheel models introduce discrepancies relative to the intractable real vehicle model and observed measurements. These discrepancies ultimately manifest as disturbances in the vehicle state. We denote this state disturbance at time step $k$ as $\bs{\xi}_k = [\xi_{x_k},\xi_{y_k},\xi_{\psi_k},\xi_{v_{x_k}},\xi_{v_{y_k}},\xi_{\omega_k}]^\top \in \mathbb{R}^{6\times 1}$. After obtaining the quasi-LPV model, the real-world state at time step $k+1$, $\bs{z}_{k+1}$, can be represented by modifying Eq. \eqref{eq:lpv_rewrite} to include the disturbance $\bs{\xi}_k$:
\begin{equation}\label{eq:single_uncertain_eq}
    \bs{z}_{k+1} = \bs{A}_k\bs{z}_k + \bs{B}_k\bs{u}_k + \bs{\xi}_k,\quad k\in\mathbb{Z}_{\geq0}.
\end{equation}
Now we reformulate Eq.~\eqref{eq:single_uncertain_eq} as follows:
\begin{equation*}
	\begin{aligned}
        \bs{z}_{k+1} &= \bs{A}_k\bs{z}_k + \bs{B}_k\bs{u}_k + \bs{\xi}_k \\
        &= \bs{A}_k(\bs{A}_{k-1}\bs{z}_{k-1} + \bs{B}_{k-1}\bs{u}_{k-1} + \bs{\xi}_{k-1}) + \bs{\xi}_k \\
        &= \bs{A}_k(\bs{A}_{k-1}(\bs{A}_{k-2}(\cdots(\bs{A}_0\bs{z}_0 + \bs{B}_0\bs{u}_0 + \bs{\xi}_0)+\bs{B}_1\bs{u}_1+ \\ 
        &\bs{\xi}_1)+\cdots)+\bs{B}_{k-1}\bs{u}_{k-1} + \bs{\xi}_{k-1}) + \bs{B}_k\bs{u}_k + \bs{\xi}_k \\
        &=\bs{E}_0^k\bs{z}_0 + \bs{E}_1^k\bs{B}_0\bs{u}_0 + \cdots + \bs{E}_k^k\bs{B}_{k-1}\bs{u}_{k-1} + \bs{B}_k\bs{u}_k +\\
        &\bs{E}_1^k\bs{\xi}_0 + \cdots + \bs{E}_k^k\bs{\xi}_{k-1} + \bs{\xi}_k,
    \end{aligned}
\end{equation*}
where \(\bs{E}_j^k = \prod_{i=j}^{k}\bs{A}_i\) represents the product operation, i.e., \(\prod_{i=j}^k \bs{A}_i = \bs{A}_k\bs{A}_{k-1}\dots \bs{A}_j\) and \(\bs{E}_k^k = \bs{A}_k\). Then we can write Eq.~\eqref{eq:single_uncertain_eq} in matrix form as follows:
\begin{equation*}
	\begin{aligned}
        \bs{z}_{k+1} &= 
        \begin{bmatrix}
            \bs{E}_0^k, \bs{E}_1^k\bs{B}_0,\cdots,\bs{E}_k^k\bs{B}_{k-1},\bs{B}_k
        \end{bmatrix}
        \begin{bmatrix}
            \bs{z}_0,\bs{u}_0,\cdots,\bs{u}_k
        \end{bmatrix}^\top\\
        &+
        \begin{bmatrix}
            \bs{E}_1^k,\bs{E}_2^k,\cdots,\bs{E}_k^k,\bs{\mathcal{I}}_{6}
        \end{bmatrix}
        \begin{bmatrix}
            \bs{\xi}_0^\top,\bs{\xi}_1^\top,\cdots,\bs{\xi}_k^\top
        \end{bmatrix}^\top.
    \end{aligned}
\end{equation*}

Considering that each iteration includes $N+1$ steps, including the initial state, i.e., $k\in \{0,1,\dots,N\}$, we denote the overall predicted states in the horizon as \(\bs{Y} = [\bs{z}_1^\top,\bs{z}_2^\top,\dots,\bs{z}_N^\top]^\top\in \mathbb{R}^{6N\times1}\) and rewrite the overall decision variables as a combination of the initial state and controls:  $\bs{X} = [\bs{z}_0^\top,\bs{u}_0^\top,\dots,\bs{u}_{N-1}^\top]^\top\in \mathbb{R}^{(6+2N)\times1}$. We can calculate $\bs{Y}$ as follows:
\begin{equation}\label{eq:LPV_lz_form}
    \bs{Y} = \bs{L}_N\bs{X} + \bs{H}_N \bs{\eta}_N,
\end{equation}
where
\begin{equation}\label{eq:Ln}
\bs{L}_N = \scalebox{0.8}{$
\begin{bmatrix}
    \bs{E}_0^0&\bs{B}_0&\bs{0}_{6\times2}&\cdots&\cdots&\bs{0}_{6\times2} \\
    \bs{E}_0^1&\bs{E}_1^1\bs{B}_0&\bs{B}_1&\bs{0}_{6\times2}&\cdots& \bs{0}_{6\times2} \\
    \vdots&\vdots&\vdots&\vdots&\ddots&\vdots \\
    \bs{E}_0^{N-1}&\bs{E}_1^{N-1}\bs{B}_0&\bs{E}_2^{N-1}\bs{B}_1&\cdots&\bs{E}_{N-1}^{N-1}\bs{B}_{N-2}&\bs{B}_{N-1}
\end{bmatrix}
$}
\end{equation}
and 
\begin{equation}\label{eq:Hn}
\bs{H}_N = \scalebox{0.8}{$\begin{bmatrix}
    \bs{\mathcal{I}}_6&\bs{0}_{6\times6}&\cdots&\cdots&\bs{0}_{6\times6} \\
    \bs{E}_1^1&\bs{\mathcal{I}}_6&\bs{0}_{6\times6}&\cdots&\bs{0}_{6\times6} \\
    \bs{E}_1^2&\bs{E}_2^2&\bs{\mathcal{I}}_6&\cdots&\bs{0}_{6\times6} \\
    \vdots&\vdots&\vdots&\ddots&\vdots \\
    \bs{E}_1^{N-1}&\bs{E}_2^{N-1}&\cdots&\bs{E}_{N-1}^{N-1}&\bs{\mathcal{I}}_6
\end{bmatrix}$}
\end{equation}
$\bs{L}_N$, $\bs{H}_N$, and $\bs{\eta}_N$ are the equivalent quasi-LPV system matrices, where $\bs{L}_N \in \mathbb{R}^{6N\times(6+2N)}$ and $\bs{H}_N \in \mathbb{R}^{6N\times6N}$. The unknown disturbance vector is defined as $\bs{\eta}_N = [\bs{\xi}_0^\top,\bs{\xi}1^\top,\dots,\bs{\xi}_{N-1}^\top]^\top \in \mathbb{R}^{6N\times1}$. Eq.~\eqref{eq:LPV_lz_form} shows that disturbances propagate and accumulate over prediction steps, underscoring the need for robust control. A sufficiently small sampling time $\Delta t$ ensures accurate quasi-LPV discretization, numerical stability, and avoids overly conservative disturbance amplification in $\bs{H}_N$.

We further represent a specific time step state by introducing a matrix $\bs{D}_k$, where $\bs{D}_k \in \mathbb{R}^{6\times6N}$:
\[
\bs{D}_k = [\underbrace{\bs{0}_{6 \times 6}, \ldots, \bs{0}_{6 \times 6}}_{k-1}, \bs{\mathcal{I}}_6, \underbrace{\bs{0}_{6 \times 6}, \ldots, \bs{0}_{6 \times 6}}_{N-k}],
\]
and then it gives:
\begin{equation}
	\begin{aligned}
    \bs{z}_{k} &= \bs{D}_k(\bs{L}_N\bs{X} + \bs{H}_N\bs{\eta}_N) \\
    &=\bs{D}_k\bs{L}_N\bs{X}+\bs{D}_k\bs{H}_N\bs{\eta}_N,\quad k\in\{1,2,\dots,N\}.\label{eq:single_state_lz}
   	\end{aligned}
\end{equation}


It is worth noting that the pure quasi-LPV model, i.e., $\bs{z}_k=\bs{D}_k\bs{L}_N\bs{X}$, cannot perfectly represent the real vehicle dynamics due to several sources of modeling discrepancy. As discussed earlier, three main sources of uncertainty are considered in this work: (1) The nonlinear STM dynamics are governed by Eq.~\eqref{eq:differential_equation}, where the slip angles are computed as
\begin{equation}
        \alpha_f = \delta - \arctan\left(\frac{\omega l_f + v_y}{v_x}\right),\quad
        \alpha_r = \arctan\left(\frac{\omega l_r - v_y}{v_x}\right).
    \label{eq:slipAngleFormulate1}
\end{equation}
The corresponding lateral tire forces follow the nonlinear Pacejka model
\begin{equation}\label{eq:nonlinearForceFormulate1}
    \begin{aligned}
        F_{fy} &= D_f \sin\left(C_f \arctan(B_f \alpha_f)\right),\\
        F_{ry} &= D_r \sin\left(C_r \arctan(B_r \alpha_r)\right),
    \end{aligned}
\end{equation}
where $B_i$, $C_i$, and $D_i$ ($i=f,r$) are the Pacejka tire coefficients. Eq.~\eqref{eq:slipAngleFormulate1} and Eq.~\eqref{eq:nonlinearForceFormulate1} provide a more accurate description of tire dynamics, while the quasi-LPV model instead uses the linear cornering-stiffness approximation in Eq.~\eqref{eq:alpha}. This linearization introduces the first source of discrepancy between the quasi-LPV model and the nonlinear STM dynamics, and further differences relative to the real vehicle (e.g., a four-wheel model).
(2) The quasi-LPV model relies on a scheduling parameter vector at the current time step to compute the system matrices. Since this parameter includes the control variable $\delta_k$, which is only available after solving the MPC optimization, it must be approximated using the previously solved decision variable. This introduces additional approximation error. 
(3) Practical implementations involve additive disturbances arising from localization inaccuracies, sensor noise, and other measurement errors. 
These discrepancies accumulate over time and introduce uncertainties in the predicted vehicle behavior. In this work, they are captured by the additive disturbance term $\xi_k$ in Eq.~\eqref{eq:single_state_lz}, which compensates for the modeling gap between the quasi-LPV model and the actual system. Because the true distribution of $\xi_k$ is generally unknown, we construct a data-driven ambiguity set using collected disturbance samples and incorporate it into the proposed DRCC controller, which is a key contribution of this paper.

In the next section, we build on Eq.~\eqref{eq:single_state_lz} to represent the obstacle avoidance constraints in a linearized form.

\section{Obstacle Avoidance via Linear Form of Constraints}
\label{sec:ObstacleAvoidance}

For the tracking problem, we enforce obstacle avoidance at the local control stage rather than during global path planning. This design choice mirrors realistic operational scenarios where many obstacles are not accessible to a global planner and are only detected online via local sensors. Therefore, to ensure that the controller maintains safety and addresses these uncertainties, we assume in both numerical and real-world experiments that obstacle positions are initially unknown to the global planner and their exact locations are provided solely by external sensors during execution.

Given this necessity for real-time, local avoidance, we must use an efficient method. For NMPC-based obstacle avoidance problems, one common approach assumes that other vehicles are represented by ellipses~\cite{muraleedharan2021real}.  This assumption allows defining the obstacle avoidance constraints as:
$
{(x_\text{obs}-x_{k})^2}/{r_x^2} + {(y_\text{obs}-y_{k})^2}/{r_y^2} \geq 1,
$
where the obstacle is represented as an ellipse with center coordinates $[x_\text{obs},y_\text{obs}]$ and axes $r_x$ and $r_y$. This kind of approach to constructing obstacle constraints involves representing obstacles with predefined geometric shapes and constraining the geometric relationship between vehicle position $[x_k,y_k]$ and the refined objects. While this method provides a reasonable description of obstacle avoidance, it introduces nonlinear constraints that are incompatible with QP. Furthermore, these nonlinear constraints increase the optimization problem's complexity and computational costs. In more complex scenarios, an excessive number of such constraints can make the NMPC problem infeasible.

Several studies~\cite{nezami2024obstacle,alcala2020lpv} have addressed this issue by linearizing obstacle constraints into a generalized linear inequality form (e.g., half-space constraints). This paper adopts a similar approach and leverages the advantages of frame conversion between the TNB frame along the tracking reference~\cite{erkan2018serret} and the Cartesian frame during the process of determining the local reference path.

In this section, we focus on deriving the linear form of the obstacle avoidance constraint and formulate the LPVMPC in a general form that will be used in the subsequent sections.

\subsection{Linear Form of Obstacle Avoidance Constraints}\label{section:linearizeObsAvoidCons}
We first analyze the obstacles in the TNB frame along the tracking reference. Based on theories such as GoalNet~\cite{zhang2021map} and DAC~\cite{narayanan2021divide}, we model the reference raceline $R_r$ in the path-relative TNB frame and denote the conversion between two frames as function $[x,y] = \Gamma_\text{r}(T,\Lambda)$, where $T$ is the tangential axis along the track, and $\Lambda$ is the perpendicular axis to $T$ in the TNB frame, it is obvious that $\Lambda = 0$ if the point is on $R_r$. Fig.~\ref{fig:FtC} illustrates the definition of obstacle constraints. For each obstacle, a bounding box is defined with its center aligned to the direction of the projected point on $R_\text{r}$. To minimize errors in estimating the scheduling parameter vector, two triangular regions are introduced at the front and rear of the bounding box, forming a trapezoidal safety region with four vertices, denoted as $[T_j,\Lambda_j]$, where $j\in \{0,A,B,1\}$ as illustrated in Fig.~\ref{fig:FtC}. The points $A$ and $B$ correspond to two of the bounding box vertices. The Cartesian coordinates of the safety region's vertices are computed as $[x_j,y_j] = \Gamma_\text{r}(T_j,\Lambda_j)$, $j\in \{0,A,B,1\}$. Given a reference point $T^\text{ref}_k$ within the safety region, we define its coordinates as $[T_k^\text{ref},\Lambda_k^\text{ref}]$ instead of using the $[T_k^\text{ref},0]$, ensuring it positions along the hypotenuse of the triangular portion of the safety region. The corresponding $\Lambda^\text{ref}_k$ is directly obtained from the vertex coordinates based on the following triangular geometry:
\begin{equation}\label{eq:TNBgetNk}
\Lambda^\text{ref}_k =
    \begin{cases}
        \frac{T_k-T_0}{T_a - T_0}\Lambda_a, & T^\text{ref}_k\in [T_0,T_A],\\
        \Lambda_a + \frac{T_k-T_a}{T_b-T_a} (\Lambda_b-\Lambda_a), & T^\text{ref}_k\in (T_A,T_B),\\
        \frac{T_1-T_k}{T_1-T_b}\Lambda_b,& T^\text{ref}_k\in [T_B,T_1].
    \end{cases}
\end{equation}
The Cartesian coordinates can then be obtained as follows:
\begin{equation}\label{eq:get_safe_reference_points}
    [x^\text{ref}_k,y^\text{ref}_k] = \Gamma_\text{r}(T^\text{ref}_k,\Lambda^\text{ref}_k).
\end{equation}

\begin{figure}[t!]
  \begin{center}
  \includegraphics[width=3.5in]{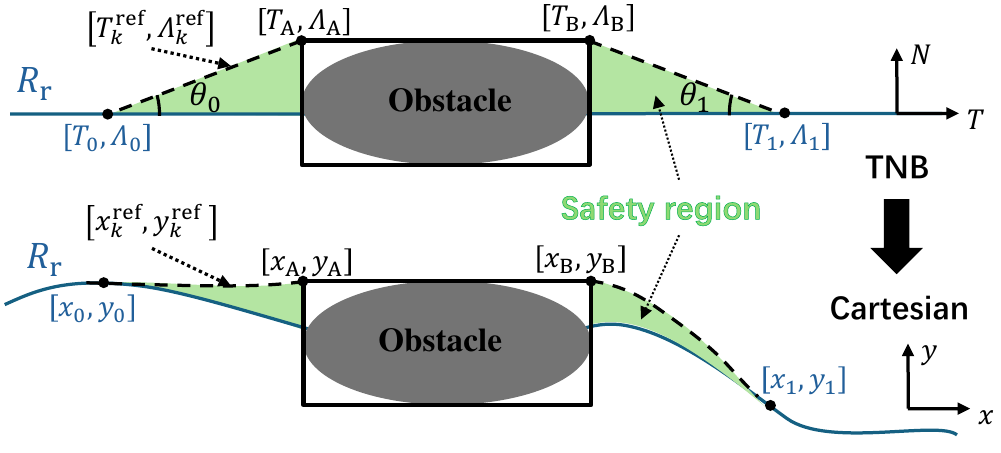}
 \vspace{-7mm}
  \caption{Safety region defined in the TNB frame and converted back to the Cartesian frame.}\label{fig:FtC}
  \end{center}
\end{figure}

\begin{figure}[t!]
  \begin{center}
  \includegraphics[width=3.5in]{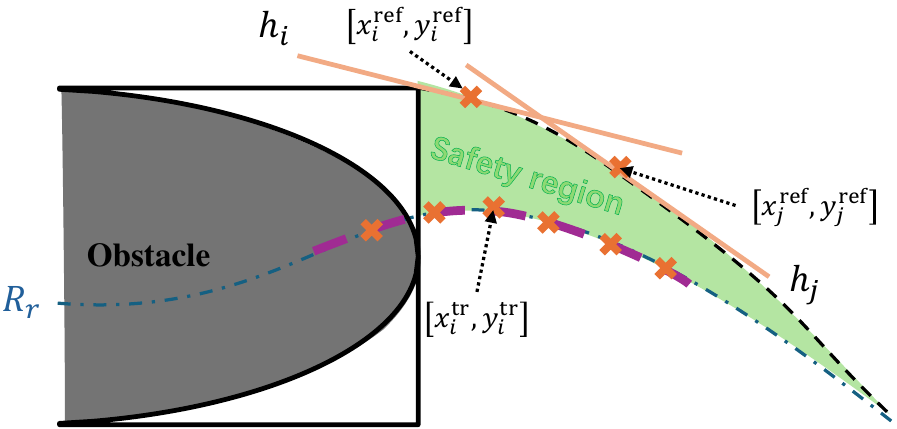}\\
   \vspace{-1mm}
  \caption{Tangent line of the obstacle avoidance constraints at time steps $i$ and $j$, where $h_i$ and $h_j$ represent the safe region boundary tangent lines at $[x^\text{ref}_i,y^\text{ref}_i]$ and $[x^\text{ref}_j,y^\text{ref}_j]$, respectively.}\label{fig:feasible_region}
  \end{center}
\end{figure}

The tangential angle $\theta_k^\text{ref}$ at the point $[x^\text{ref}_k,y^\text{ref}_k]$ can be calculated by:
\begin{equation}\label{eq:theta_reference}
\theta_k^\text{ref} =
    \begin{cases}
        \theta_k^\text{tr} + \theta_0, & T^\text{ref}_k\in [T_0,T_A], \\
        \theta_k^\text{tr} - \theta_1, & T^\text{ref}_k\in [T_B,T_1],
    \end{cases}
\end{equation}
where $\theta_0$ and $\theta_1$ are the vertex angles shown in Fig.~\ref{fig:FtC}. Thus, for a given reference point $[x^\text{ref}_k,y^\text{ref}_k]$ from Eq.~\eqref{eq:get_safe_reference_points} and the tangential angle from Eq.~\eqref{eq:theta_reference}, the equation of the corresponding tangent line $h_k$ can be obtained, as shown in Fig.~\ref{fig:feasible_region}. Here, $[x_i^\text{tr},y_i^\text{tr}]=\Gamma_\text{r}(T_i^\text{ref},0)$ denotes the original reference point where $\Lambda$ is zero, while $[x_k^\text{ref},y_k^\text{ref}]$ represents the shifted reference point with the same $T^\text{ref}$ but an updated $\Lambda^\text{ref}$, obtained using Eq.~\eqref{eq:TNBgetNk}. Ideally, the feasible region for $[x_k,y_k]$ should lie outside the safety region. Depending on the required vehicle maneuver to avoid obstacles, the constraint is formulated as:
\begin{equation}\label{eq:obs_linear_constraint}
h_k = 
    \begin{cases}
        a_k x_k + b_k y_k - c_k \geq 0,&\text{Left Turn}, \\
        a_k x_k + b_k y_k - c_k \leq 0,&\text{Right Turn},
    \end{cases}
\end{equation}
where $a_k$, $b_k$, and $c_k$ are the coefficients of the tangent line. The decision on which side the vehicle should take to bypass obstacles can be determined by topology-based planners, such as~\cite{de2023globally}, and is beyond the scope of this work.

\subsection{LPVMPC}
In the case of pure LPVMPC, where no disturbance is considered, the state at the $k$th time step, based on Eq.~\eqref{eq:single_state_lz}, is given by:
\begin{equation}\label{eq:lpv_pure_state}
    \bs{z}_k = \bs{D}_k\bs{L}_N\bs{X},\quad k\in\{1,2,\dots,N\}.
\end{equation}
By introducing the vector $\bs{d}_k = [a_k,b_k,0,0,0,0]$, we can reformulate Eq.~\eqref{eq:obs_linear_constraint} using Eq.~\eqref{eq:lpv_pure_state}, leading to:
\begin{equation}\label{eq:obs_linear_constraint_standard_form}
h_k = 
    \begin{cases}
        \bs{d}_k\bs{D}_k\bs{L}_N\bs{X} - c_k \geq 0,&\text{Left Turn}, \\
        -\bs{d}_k\bs{D}_k\bs{L}_N\bs{X} - (-c_k) \geq 0,&\text{Right Turn}.
    \end{cases}
\end{equation}

Besides obstacle avoidance, all linear constraints can be expressed in the form of Eq.~\eqref{eq:obs_linear_constraint_standard_form} by appropriately selecting $\bs{d}_k$ and $c_k$. For instance, the maximum velocity constraint can be represented by setting $\bs{d}_k = [0,0,0,1,0,0]$ and $c_k = v_{x_\text{max}}$, etc. Assuming there are $Q$ types of constraints in total, the LPVMPC formulation can be written as:
\begin{align}\label{eq:LPVMPC_formulation}
    \min_{\bs{X}}& \sum_{k=1}^N \{P_k[(\bs{e}_0^6\bs{D}_k\bs{L}_N\bs{X}-x^\text{ref}_k)^2 + (\bs{e}_1^6\bs{D}_k\bs{L}_N\bs{X} - y^\text{ref}_k)^2]\}\nonumber\\ 
    \text{s.t.}& \quad \bs{d}_k^i\bs{D}_k\bs{L}_N\bs{X} -c_k^i \geq 0, k\in [N], i\in[Q], \quad \cdots 
\end{align}
where $P_k$ represents the tracking weights at the $k$th time step, and `$\dots$' denotes additional constraints, such as control input limits, initial conditions, and other system-specific restrictions. $x^\text{ref}_k$ and $y^\text{ref}_k$ are obtained based on Eq.~\eqref{eq:get_safe_reference_points}. Near obstacles, we select $\Lambda_k^\text{ref}$ using Eq.~\eqref{eq:TNBgetNk} as the shifted distance to the track. When no obstacles are present, we set $\Lambda_k^\text{ref} =0$.

\section{DRCC-LPVMPC}
\label{sec:DRCC-LPVMPC}
This section presents the DRCC-LPVMPC formulation, incorporating the linearized obstacle constraints derived earlier and accounting for additive uncertainty to improve robustness.
\subsection{CVaR Reformulation}

Eq.~\eqref{eq:LPVMPC_formulation} gives a form of hard constraint. When considering uncertainty with Eq.~\eqref{eq:single_state_lz}, the constraint is modified as:
\begin{equation}\label{eq:DRCCLPVMPC_ccps}
    \bs{d}_k^i\bs{D}_k\bs{L}_N\bs{X} + \bs{d}_k^i\bs{D}_k\bs{H}_N\bs{\eta}_N -c_k^i \geq 0, k\in [N], i\in[Q]
\end{equation}
we denote $\tilde{\bs{\eta}}_N$ as the unknown distribution of the uncertainty and denote $\tilde{w}^{ki} = -\bs{d}_k^i\bs{D}_k\bs{H}_N\tilde{\bs{\eta}}_N$ as the refined distribution of the uncertainty for the $i$th type of constraint at the $k$th time step. Since we have already converted the $Q$ types of constraints into the same form as Eq.~\eqref{eq:DRCCLPVMPC_ccps}, for simplicity, we will omit the index $i$ for the remainder of this section and only focus on obstacle avoidance constraint, i.e., $\bs{d}^i_k =\bs{d}_k = [a_k,b_k,0,0,0,0]$, $c^i_k = c_k$, and $\tilde{w}^{ki} = \tilde{w}^k$. Thus, Eq.~\eqref{eq:DRCCLPVMPC_ccps} simplifies to: 
$\bs{d}_k\bs{D}_k\bs{L}_N\bs{X} -c_k \geq \tilde{w}^k,\quad k\in [N]$.

The DRCC requires that constraints hold with probability at least $1-\epsilon$ for all distributions $\mathbb{P}$ within the ambiguity set $\mathcal{P}$, also called the ``Wasserstein ambiguity set''. Here, $\epsilon \in (0,1)$ is a specified risk tolerance. Formally, for the decision variable $\bs{X}$, the DRCC is expressed as
\begin{equation}\label{eq:DRCCPs}
    \inf_{\mathbb{P}\in \mathcal{P}} \mathbb{P}\{\bs{d}_k\bs{D}_k\bs{L}_N\bs{X} -c_k \geq \tilde{w}^k, \forall k\in [N]\} \geq1-\epsilon,
\end{equation}
where $\mathbb{P}$ represents a joint probability distribution from a family $\mathcal{P}$. As shown in~\cite{fournier2015rate}, $\mathcal{P}$ defines the Wasserstein distance between the true and the empirical state distributions, which ensures the true distribution is recovered as the sample size approaches infinity. Eq.~\eqref{eq:DRCCPs} represents an infinite-dimensional constraint because it requires the safety constraint to hold for all probability distributions within the ambiguity set. Such constraints cannot be solved directly. Therefore, the DRCC constraint is reformulated using a CVaR representation, which converts it into a finite-dimensional convex constraint that can be efficiently solved while preserving the desired probabilistic safety guarantee. Equivalently, this can be written as
\begin{equation}\label{eq:DRCCPs2}
    \sup_{\mathbb{P}\in \mathcal{P}} \mathbb{P}(\bs{d}_k \bs{D}_k \bs{L}_N\bs{X} - c_k < \tilde{w}^k, \exists k\in [N]) \leq \epsilon.
\end{equation}
The probabilities can be expressed in terms of expectations: $\mathbb{P}(\bs{d}_k\bs{D}_k\bs{L}_N\bs{X}-c_k <\tilde{w}^k,\exists k \in [N]) =\mathbb{E}_\mathbb{P}[\mathbb{I}(\bs{d}_k\bs{D}_k\bs{L}_N\bs{X}-c_k<\tilde{w}^k,\exists k \in [N])]$,
where the indicator function $\mathbb{I}(\cdot)$ can be written as:
\begin{align*}
&\mathbb{I}(\bs{d}_k\bs{D}_k\bs{L}_N\bs{X}-c_k<\tilde{w}^k,\exists k\in[N]) \nonumber\\
&=\begin{cases}
    1,&\bs{d}_k\bs{D}_k\bs{L}_N\bs{X}-c_k <\tilde{w}^k,\\
    0,&\text{otherwise}.
    \end{cases}
\end{align*}
This shows that Eq.~\eqref{eq:DRCCPs} and Eq.~\eqref{eq:DRCCPs2} are both bounded and upper semi-continuous. Assume a total of $J \in \mathbb{Z}_{>0}$ samples are drawn from the true distribution $\mathbb{P}^\infty$. Following Theorem~1 in~\cite{gao2023distributionally} and Theorem~1 in~\cite{blanchet2019quantifying}, the left side of Eq.~\eqref{eq:DRCCPs2} is equivalent to 
\begin{align}\label{eq:DRCCPs3}
    \min_{\lambda_k \geq0}\{&\lambda_k\sigma - \frac{1}{N}\sum_{j\in[J]} \inf_{\tilde{w}}[\lambda_k\|\tilde{w}-\hat{w}_j\|- \nonumber\\
    & \mathbb{I}(\bs{d}_k\bs{D}_k\bs{L}_N\bs{X}-c_k<\tilde{w}^k,\exists k\in[N])]\},
\end{align}
where $\hat{w}_j$, $j\in [J]$, represents the $j$th set of refined sampled disturbance, as mentioned in the notation. $\lambda_k$ and $\sigma$ are dimensionless variables. Specifically, $\lambda_k>0$ is introduced to support the reformulation of the DRCC, while $\sigma>0$ denotes the Wasserstein radius of $\mathcal{P}$.

According to~\cite{xie2021distributionally}, starting from Eq.~\eqref{eq:DRCCPs3}, we first examine the infinite part: $  \inf_{\tilde{w}}[\lambda_k\|\tilde{w}-\hat{w}_j\|- \mathbb{I}(\bs{d}_k\bs{D}_k\bs{L}_N\bs{X}-c_k<\tilde{w}^k,\exists k\in[N])]$,
which is equivalent to
\begin{equation}\label{eq:DRCCPs4}
\min\{\min_{k\in[N]}\quad \inf_{\bs{d}_k\bs{D}_k\bs{L}_N\bs{X}-c_k<\tilde{w}^k}[\lambda_k\|\tilde{w}-\hat{w}_j\|-1],0\}.
\end{equation}
Now, we denote
\[
f(\bs{X},\hat{w}_j) = \min\{\min_{k\in[N]}\max\{\bs{d}_k\bs{D}_k\bs{L}_N\bs{X}-c_k-\hat{w}_j^k,0\},0\},
\]
and use the fact that
\begin{equation}\label{eq:DRCCPfact}
\inf_{\bs{d}_k\bs{D}_k\bs{L}_N\bs{X}-c_k<\tilde{w}^k}\|\tilde{w}-\hat{w}_j\| = \max\{\bs{d}_k\bs{D}_k\bs{L}_N\bs{X}-c_k-\hat{w}_j^k,0\}.
\end{equation}
We obtain the following result:
\begin{theorem}\label{theorem:1}
The DRCC problem in Eq.~\eqref{eq:DRCCPs} is equivalent to
\begin{equation*} 
    \lambda_k\sigma-\epsilon\leq\frac{1}{N}\sum_{j\in[J]}\min\{\lambda_k f(\bs{X},\hat{w}_j)-1,0\}.
\end{equation*}
\end{theorem}
\begin{proof}
See Appendix~\ref{app1}.
\end{proof}

Based on Theorem~\ref{theorem:1}, let $\gamma_k = \frac{1}{\lambda_k}$, the feasible region of DRCC, denoted as $\mathbb{X}_{D}$, then becomes:
\small
\begin{equation}\label{eq:DRCCPs7}
    \mathbb{X}_{D} = \left\{ \bs{X} \in \mathbb{R}^{6+2N} :
    \begin{aligned}
        & \sigma - \epsilon \gamma_k \leq \frac{1}{N} \sum_{j \in [J]} 
        \min \{ f(\bs{X}, \hat{w}_j) - \gamma_k, 0 \}, \\
        & \gamma_k \geq 0,\forall k \in [N],
    \end{aligned}
    \right\}
\end{equation}
\normalsize
which is equivalent to
\begin{equation*}
\sigma-\epsilon\gamma_k-\frac{1}{N}\sum_{j\in[J]}\max\{-f(\bs{X},\hat{w}_j)+\gamma_k,0\}\leq0.
\end{equation*}
Let $\varpi_k = -\gamma_k$, then
\begin{equation*}
\frac{\sigma}{\epsilon}+\varpi_k-\frac{1}{N\epsilon}\sum_{j\in[J]}\max\{-f(\bs{X},\hat{w}_j)-\varpi_k,0\}\leq0,
\end{equation*}
which can be written as:
\begin{equation}\label{eq:DRCCPs8}
\frac{\sigma}{\epsilon}+\min_{\varpi_k}\{\varpi_k+\frac{1}{N\epsilon}\sum_{j\in[J]}\max\{-f(\bs{X},\hat{w}_j)-\varpi_k,0\}\}\leq0.
\end{equation}
Note the definition of CVaR is:
\[
\text{CVaR}_{1-\epsilon}[-f(\bs{X},\tilde{w})] = \min_{\varpi_k}\{\varpi_k+\frac{1}{\epsilon}\mathbb{E}_{\mathbb{P}\tilde{w}}[-f(\bs{X},\tilde{w})-\varpi_k]_t\}.
\]
Thus, Eq.~\eqref{eq:DRCCPs8} can be written in the following CVaR form:
\begin{equation}\label{eq:DRCC_cvar}
    \mathbb{X}_D=\{\bs{X}\in\mathbb{R}^{6+2N}:\frac{\sigma}{\epsilon}+\text{CVaR}_{1-\epsilon}[-f(\bs{X},\tilde{w})]\leq0\},
\end{equation}
and the concept of CVaR is illustrated in Fig.~\ref{fig:cvaR_depicts}, where $-f(\bs{X},\tilde{w})$ is the function and $\epsilon$ refers to the risk tolerance.

\begin{figure}
  \begin{center}
  \includegraphics[width=3.4in]{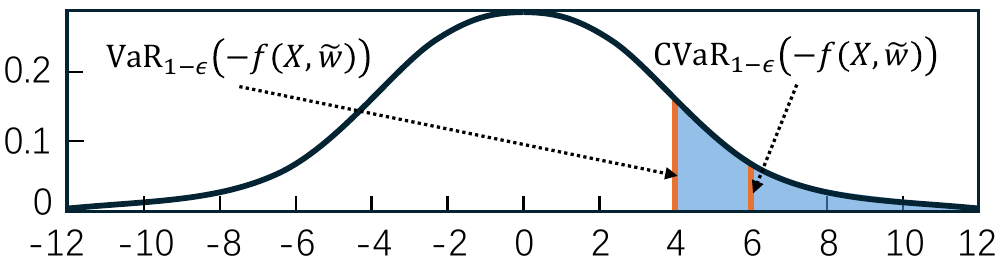}
     \vspace{-3mm}
  \caption{Depiction of CVaR for a function with an unknown distribution (illustrated as Gaussian for convenience). The Value-at-Risk (VaR) represents the threshold below which a specified percentage ($\epsilon$) of the distribution lies, while CVaR accounts for the average loss beyond this threshold, effectively capturing tail risk. The shaded region accounts for $\epsilon \times 100\%$ of the mass of the unknown distribution.}\label{fig:cvaR_depicts}
  \end{center}
\end{figure}

\subsection{Mixed-Integer Program Reformulation}
Based on the CVaR formulation, this section derives the mixed-integer reformulation of the constraints to ensure tractability while maintaining robustness.
\begin{theorem}\label{theorem:2}
The constraints in Eq.~\eqref{eq:DRCCPs7}
\begin{subequations}\label{eq:DRCCP9}
    \begin{align}
        & \sigma - \epsilon \gamma_k \leq\frac{1}{N}\sum_{j\in[J]}\min\{f(\bs{X},\hat{w}_j)-\gamma_k,0\}, \label{eq:DRCCP9_1}\\
        & \gamma_k \geq 0,
    \end{align}
\end{subequations}
is equivalent to $\forall j\in[J]$, $\forall k\in[N]$,
\begin{subequations}\label{eq:DRCCP10}
    \begin{align}
        & \sigma - \epsilon \gamma_k \leq\frac{1}{N}\sum_{j\in[J]}q_j^k, \label{eq:DRCCP10_1}\\
        & q_j^k+\gamma_k\leq\max\{\bs{d}_k\bs{D}_k\bs{L}_N\bs{X}-c_k-\hat{w}_j^k,0\}, \label{eq:DRCCP10_2}\\
        & q_j^k\leq0,\quad \gamma_k \geq 0.\nonumber
    \end{align}
\end{subequations}
\end{theorem}
\begin{proof}
See Appendix~\ref{app2}.
\end{proof}

Based on ~\cite{xie2021distributionally}, Eq.~\eqref{eq:DRCCPs} can be written as
\[
\inf_{\mathbb{P}\in\mathcal{P}}\mathbb{P}\{\bar{\bs{B}}_k(\bs{X})\bs{X}+\bar{b}^k \geq \mathbf{a}(\bs{X})^\top\tilde{e}_w,\forall k\in[N]\} \geq 1-\epsilon,
\]
where $\bar{\bs{B}}_k(\bs{X}) := \bs{d}_k\bs{D}_k\bs{L}_N$, $\bar{b}^k:=-c_k$, and $\mathbf{a}(\bs{X})^\top:=\bs{e}_7^7$. 
This implies the DRCC problem with right-hand uncertainty, where all the elements of $\mathbf{a}(\bs{X})$ are zero except for the last element, which equals one. From this we can deduce the following corollary:
\begin{corollary}\label{corollary1}
    The set $\mathbb{X}_D$ is mixed-integer representable by introducing a sufficiently large set of numbers $M_j\in \mathbb{R}_+,j\in[J]$, i.e.:
    \small
    \begin{equation}\label{eq:MjFormulation}
    \max_{k\in [N]}\max_{\bs{X}\in \mathbb{X}_D}\{|\bs{d}_k\bs{D}_k\bs{L}_N\bs{X}-c_k-\hat{w}_j^k|\}\leq M_j,\forall j\in[J],\forall k\in [N].
    \end{equation}
    \normalsize
    The mixed-integer representation of $\mathbb{X}_D$ is:  for all $j \in [J]$ and $\forall k\in[N]$,
    \small
    \begin{subequations}\label{eq:DRCCP11}
        \begin{align}
            & \sigma - \epsilon \gamma_k \leq\frac{1}{N}\sum_{j\in[J]}q_j^k, \quad q_j^k+\gamma_k\leq s_j^k,\\ 
            & s_j^k\leq \bs{d}_k\bs{D}_k\bs{L}_N\bs{X}-c_k-\hat{w}_j^k+M_j(1-r_j^k), \label{eq:DRCCP11_3}\\
            & s_j^k\leq M_jr_j^k, \label{eq:DRCCP11_4}\\
            & \gamma_k \geq 0,\quad q_j^k\leq0, \quad s_j^k\geq0,\quad r_j^k\in\{0,1\},\nonumber
        \end{align}
    \end{subequations}
    \normalsize
\end{corollary}
\begin{proof}
See Appendix~\ref{app3}.
\end{proof}
\begin{remark}
    In contrast to the DRO framework, which optimizes the ``worst-case'' distribution over all possible distributions in $\mathcal{P}$, the proposed approach seeks the ``worst-case'' robust control with respect to the probability distribution inferred from sampled data. This is the sense in which our method is ``data-driven distributionally robust''. Furthermore, Eq.~\eqref{eq:DRCCP11} reformulates the infinite chance constraints in Eq.~\eqref{eq:DRCCPs} into first-order constraints by introducing dimensionless variables $q_j^k$, $\gamma_k$, $s_j^k$, and $r_j^k$. This reformulation enables the problem to be cast as a QP, making it more tractable for computation.
\end{remark}

Fig.~\ref{fig:DRCC_summary} shows an overall summary of how the infinite-dimensional DRCC is converted into convex constraints. The data-driven method leverages a sampled data set $\hat{w}_j$ to represent an empirical distribution that satisfies the Wasserstein ambiguity set $\mathcal{P}$, providing a computational trade-off.
\begin{figure}
    \centering
    \includegraphics[width=0.69\linewidth]{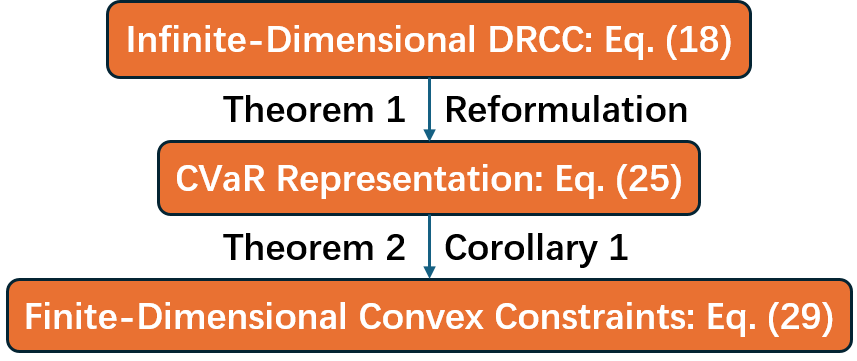}\vspace{-3mm}
    \caption{Summary of the conversion process from infinite-dimensional DRCC to convex constraints using a data-driven approach.}
    \label{fig:DRCC_summary}
\end{figure}

\subsection{Terminal Constraints}

Before introducing the terminal constraints, we define the expected state along the horizon as $\bar{\bs{z}}_{t+k} = \mathbb{E}_{\mathbb{P}}\{\bs{z}_{t+k}\}, k=0,1,...,N,$ where $\mathbb{P}\in\mathcal{P}$ is defined in Eq.~\eqref{eq:DRCCPs}. Assuming the disturbances of the ambiguity set $\mathcal{P}$ satisfy $\mathbb{E}_\mathbb{P}\{\bs{\xi}_{t+k}\} = \bs{0}_6$, the expected system dynamics in Eq.~\eqref{eq:single_uncertain_eq} become:
\begin{equation}\label{eq:ExpectSystem}
    \bar{\bs{z}}_{t+k+1} = \bs{A}_{t+k} \bar{\bs{z}}_{t+k} + \bs{B}_{t+k}\bs{u}_{t+k}.
\end{equation}
In a more compact form, this is: $\bar{\bs{z}}_{t+1} = \bs{A}_{t} \bar{\bs{z}}_{t} + \bs{B}_{t}\bs{u}_{t}.$

Terminal constraints are critical for guaranteeing system stability and performance in MPC~\cite{mayne2000survey}. Typically, for an autonomous system whose configuration satisfies
$\det(\bs{A}_t) \neq 0$, where $\det(\cdot)$ denotes the determinant of a matrix, the following terminal constraints are imposed on the nominal state: $
\bar{\bs{z}}_{t+N} \in \mathcal{Z}_T$,
where $\mathcal{Z}_T$ is a positive invariant set satisfying
\begin{equation}\label{eq:terminalConstraints}
    \bs{A}_t \bar{\bs{z}} \in \mathcal{Z}_T, \forall \bs{z} \in \mathcal{Z}_T.
\end{equation}
Based on Eq.~\eqref{eq:MjFormulation}, the system constraints along the horizon are reformulated in terms of expected states as
\begin{equation}\label{eq:systemConst}
    \bs{d}_{k}\bar{\bs{z}}_k-c_k\leq M_j,\quad\forall j\in[J],\quad\forall k\in [N],
\end{equation}
where $\bs{d}_k$ is defined in Eq.~\eqref{eq:DRCCLPVMPC_ccps}. In addition, for all $\bs{z} \in \mathcal{Z}$ and $j \in [J]$, the following condition must hold:
\begin{equation}\label{Eq:systemBound}
    \bs{d}\bar{\bs{z}} \leq M_j -c.
\end{equation}
\begin{remark} 
    Eq.~\eqref{eq:ExpectSystem} is derived by assuming that the expectation of the system disturbance is zero. This is a common assumption in DRCC problems, e.g., \cite{li2021distributionally}, even when disturbances are unbounded. In our proposed approach, if the disturbance expectation is not zero, Eq.~\eqref{Eq:systemBound} can still be satisfied by appropriately shifting the values of $M_j$. As shown in the proof of Theorem~\ref{theorem:RecursiveFeasibility}, both Eq.~\eqref{eq:terminalConstraints} and Eq.~\eqref{Eq:systemBound} are essential to guarantee the recursive feasibility of the proposed method.
\end{remark}

\subsection{MPC Formulation and Recursive Feasibility}
After obtaining Eq.~\eqref{eq:DRCCP11}, we can now formulate the DRCC-LPVMPC as follows:

\begin{align}\label{DRCCLPVMPC_formulation}
    \min_{\bs{X}} \sum_{k=1}^N \{&P_k[(\bs{e}_0^6\bs{D}_k\bs{L}_N\bs{X}-x^\text{ref}_k)^2 + (\bs{e}_1^6\bs{D}_k\bs{L}_N\bs{X} - y^\text{ref}_k)^2] \nonumber\\
    &+ \rho_k\sum_{j\in[J]} (q_j^k+\gamma_k)\}\nonumber\\
    \text{s.t.}&\quad Eq.~\eqref{eq:DRCCP11},\forall i \in [Q],\forall j\in [J], \forall k\in[N], \cdots
\end{align}
where $P_k$ and $\rho_k$ denote the weights for the tracking and DRCC penalties, respectively. The DRCC penalty is added to tighten the constraints and prevent suboptimal solutions. Without this penalty term, the condition $q_j^k+\gamma_k \leq s_j^k$ could be trivially satisfied by selecting a sufficiently small $q_j^k + \gamma_k$. The ``$\dots$'' represents other hard constraints that do not need to be converted into DRCC form, such as the range of control inputs, initial conditions, etc.

\begin{theorem}\label{theorem:RecursiveFeasibility}
   If there exists a feasible solution at time $t=0$ and the determinant of the system matrix $\det (\bs{A}_k) \neq 0$, then DRCC-LPVMPC is recursively feasible.
\end{theorem}
\begin{proof}
See Appendix~\ref{app4}.
\end{proof}
The proof shows that given feasibility at time $t$, a feasible sequence at $t+1$ can be constructed using the previous optimal control sequence and terminal constraints. By verifying that both system and terminal constraints are satisfied, we conclude that DRCC-LPVMPC is recursively feasible.

After reformulation, Eq.~\eqref{DRCCLPVMPC_formulation} leads to a re-linearized QP scheme that is essentially equivalent to the real-time iteration (RTI) method in~\cite{diehl2002real}. In that method, a multiple-shooting technique is used to discretize the nonlinear continuous-time optimal control problem, an exponential map is applied as the integrator, and only the first sequential QP step is computed at each sampling instant. Hence, the convergence properties of both Eq.~\eqref{eq:LPVMPC_formulation} and Eq.~\eqref{DRCCLPVMPC_formulation} follow directly from the known RTI convergence guarantees~\cite{diehl2005real}.

\begin{remark}
$\det (\bs{A}_k)$ can be calculated based on the vehicle's specifications and the preset state and control limits. For example, in the numerical experiments below, $\bs{A}_k$ defined in Eq.~\eqref{eq:lpv_rewrite} is positive definite as calculated using the data in Table~\ref{tab:car_configure} and Table~\ref{tab:AppliedDisturbance}. This ensures that the state predictions remain within the same invariant set. 
\end{remark}

\subsection{Discussion on Tuning Parameters}
In practice, $\sigma$, $\epsilon$, and $\rho$ are tuning parameters that balance feasibility, conservatism, and robustness. The Wasserstein radius $\sigma$ determines how tight or loose the ambiguity set is relative to the empirical distribution: smaller values yield solutions closer to the empirical data but may reduce feasibility, while larger values enlarge the set and provide greater robustness at the cost of conservatism. The risk tolerance $\epsilon$ specifies the acceptable probability of constraint violation, with smaller values enforcing stricter constraint satisfaction but shrinking the feasible region. Finally, the penalty weight $\rho$ regulates how conservatively chance constraints are enforced by enlarging safety margins; higher values prevent overly optimistic or degenerate solutions and ensure more meaningful and reliable performance.

\section{Numerical Experiment Setup}
\label{sec:Setup}
To evaluate the performance of the proposed DRCC-LPVMPC, we conducted numerical simulations comparing DRCC-LPVMPC, LPVMPC, and NMPC for a racing-track trajectory-tracking task. All controllers were tested under identical vehicle, environment, and reference configurations. Two conditions were considered: undisturbed and disturbed.
In the undisturbed condition, the experiments assumed the current model deviation was zero in LPVMPC, and DRCC-LPVMPC relied only on historical sample data calculated by model deviation, i.e., the measurement errors are zero. In the disturbed condition, measurement errors are no longer assumed to be zero and are introduced using random disturbances to reflect the real-world scenarios.
All simulations were performed on a 12th-generation Intel(R) i7-12700K CPU with 32 GB RAM, and QP problems were solved using the ProxQP solver~\cite{bambade2022prox}. All simulations use the same 1:43 scaled testbed, shown in Fig.~\ref{fig:track_lpvmpc} from~\cite{liniger2015optimization}, with the green curve representing the reference raceline for tracking, generated using B-splines from the parameterized track centerline. The positions of obstacles (black boxes) are unknown to the vehicle from the start, and safety regions (Section~\ref{section:linearizeObsAvoidCons}) shown as black dashed lines.
Both LPVMPC and DRCC-LPVMPC in this work utilize the same 1:43 scaled Autonomous RC Racing (ORCA) platform from ETH Z\"{u}rich. The vehicle configuration is listed in Table~\ref{tab:car_configure}.

\begin{table}[b!]
\addtocounter{table}{1}
	\caption{Dynamic model parameters of the 1:43 scaled ORCA car}
	\vspace{-2mm}
	\label{tab:car_configure}
	\renewcommand{\arraystretch}{1.2} 
	\begin{tabularx}{\columnwidth}{|c X|c X|}
		\hline
		\textbf{Parameter} & \textbf{Value} & \textbf{Parameter} & \textbf{Value} \\ \hline
		$l_f$           & 0.029 m         & $C_{\alpha r}$     & 2.24           \\ \hline
		$l_r$           & 0.033 m         & $v_{x_{\min}}$  & 1.2 m/s         \\ \hline
		$m$            & 0.041 kg         & $v_{x_{\max}}$  & 1.5 m/s         \\ \hline
		$l_w$           & 0.03 m          & $v_{y_{\min}}$  & -0.5 m/s        \\ \hline
		$I_z$ & $27.8 \times 10^{-6}$ $\text{kg·m}^2$ & $v_{y_{\max}}$  & 0.5 m/s \\ \hline
		$B_f$              & 2.579          & $\omega_{\min}$  & -20.94 rad/s    \\ \hline
		$C_f$              & 1.2            & $\omega_{\max}$  & 20.94 rad/s     \\ \hline
		$D_f$              & 0.192          & $a_{c_{\min}}$ & -0.4 $\text{m/s}^2$   \\ \hline
		$B_r$              & 3.3852         & $a_{c_{\max}}$ & 0.4 $\text{m/s}^2$    \\ \hline
		$C_r$              & 1.2691         & $\delta_{\min}$ & -0.59 rad      \\ \hline
		$D_r$              & 0.1737         & $\delta_{\max}$ & 0.59 rad       \\ \hline
		$C_{\alpha f}$     & 1.78           &                     &              \\ \hline
	\end{tabularx}
\end{table}

\subsection{LPVMPC Experiment Strategy}\label{section:LPVMPC}
As discussed in Section~\ref{section:linearizeObsAvoidCons}, this work calculates the local reference path in each iteration by converting between the TNB frame and the Cartesian frame. The points on $R_\text{r}$ are represented by the following features: $[T,\Lambda,\theta,S]$, where $T$ and $\Lambda$ are defined in Section~\ref{section:linearizeObsAvoidCons}, $\theta$ is the tangential direction at each point, and $S$ is the travel distance along $R_\text{r}$ up to the current point. Assume the horizon (maximum number of time steps) for each iteration is $N$. Given the reference raceline $R_\text{r}$, we denote $\Gamma^T_\text{in}$, $\Gamma^S_\text{in}$, and $\Gamma^\theta_\text{in}$ as its transform functions for calculating $T$, $S$, and $\theta$, respectively. 
The local reference points $[x^\text{ref},y^\text{ref}]$ can be obtained using Algorithm~\ref{alg:local_reference}.

\begin{algorithm}[hb!]
\caption{Local Reference Calculation}
\label{alg:local_reference}
\DontPrintSemicolon
\SetAlgoVlined
\SetKwInOut{Input}{Input}
\SetKwInOut{Output}{Output}
\Input{$R_\text{r}$, $[x_0,y_0,v_{x_0}]$, $N$}
\Output{$[x^\text{ref},y^\text{ref},\theta^\text{ref},v_x^\text{ref}]$}
\nl $T_0 \gets \Gamma_\text{in}^T(x_0,y_0)$, $S_0 \gets \Gamma_\text{in}^S(T_0)$\;
\For{$ k = 1$ to $N+1$}{ 
    \nl $v_{x_k}^\text{ref} \gets \min(v_{x_{k-1}}+a_{c_\text{max}}\Delta t,v_{x_\text{max}})$\;
    \nl $S_k^\text{ref} \gets S_{k-1}+v_{x_k}^\text{ref}\Delta t$, $T_k^\text{ref} \gets \Gamma_\text{r}^T(S^\text{ref})$\;
    \nl \lIf{$T_k^\text{ref} \in [T_0,T_1]$}{
         $\Lambda^\text{ref}_k \gets $ Eq.~\eqref{eq:TNBgetNk}}
    \nl \lElse{
        $\Lambda^\text{ref}_k \gets 0$
    }
    \nl $[x_k^\text{ref},y_k^\text{ref}] \gets \Gamma_\text{r}(T_k^\text{ref},\Lambda_k^\text{ref})$, $\theta_k^\text{ref} \gets \Gamma_\text{in}^{\theta}(T_k^\text{ref})$
}
\end{algorithm}

Algorithm~\ref{alg:local_reference} begins by calculating the initial position's $T_0$ and $S_0$ from $[x_0,y_0]$ (line 1). It then iteratively computes the longitudinal velocity and distance, as well as $T^\text{ref}$ (line 2-3). Next, $\Lambda^\text{ref}$ is determined based on Eq.~\eqref{eq:TNBgetNk} (lines 4-5). Finally, the Cartesian coordinates of the local reference $[x^\text{ref},y^\text{ref}]$ and $\theta^\text{ref}$ are calculated using the obtained $[T^\text{ref},\Lambda^\text{ref}]$ (line 6).

In the racing track scenario of this work, in addition to obstacle avoidance, it is also necessary to ensure the vehicle stays within the track boundaries. A common approach is to assume that the vehicle runs within a corridor, where both the upper and lower boundaries of the corridor can be formulated using a method similar to the one discussed in Section~\ref{section:linearizeObsAvoidCons}. For more details, we refer the reader to~\cite{liniger2015optimization}. In this work, we denote the upper and lower corridor boundary constraints, along with the obstacle avoidance constraints, as follows:
\begin{subequations}\label{eq:LPVSumupConstraints}
	\begin{align}
		&\bs{d}_k^u\bs{D}_k\bs{L}_N\bs{X}-c_k^u \geq 0, \forall k\in [N],\label{eq:LPVSumupConstraints1}\\
		&\bs{d}_k^l\bs{D}_k\bs{L}_N\bs{X}-c_k^l \geq 0,\forall k \in [N], \label{eq:LPVSumupConstraints2}\\
		&\bs{d}_k^o\bs{D}_k\bs{L}_N\bs{X}-c_k^o \geq 0,\forall k \in [N], \label{eq:LPVSumupConstraints3}
	\end{align}
\end{subequations}
where $\bs{d}_k^u$, $c_k^u$, $\bs{d}_k^l$, and $c_k^l$ are the coefficients of the upper and lower tangential lines, respectively, and $\bs{d}_k^o$, $c_k^o$ are the coefficients defined by Eq.~\eqref{eq:obs_linear_constraint_standard_form}. 
We can further represent the initial constraint as follows:
\begin{equation}\label{eq:lpvmpcx0}
    \bs{d}_0\bs{X} = \bs{z}_0,
\end{equation}
where $\bs{d}_0 = [\bs{1}_{1\times6},\bs{0}_{1\times2N}] \in \mathbb{R}^{1\times(6+2N)}$. Then, the LPVMPC in Eq.~\eqref{eq:LPVMPC_formulation} can be written as follows:
\begin{align}\label{eq:LPVMPC_experimental_formulation}
    \min_{\bs{X}}& \sum_{k=1}^N \{P_k[(\bs{e}_0^6\bs{D}_k\bs{L}_N\bs{X}-x^\text{ref}_k)^2 + (\bs{e}_1^6\bs{D}_k\bs{L}_N\bs{X} - y^\text{ref}_k)^2]\}\nonumber\\
    \text{s.t.}& \quad Eq.~\eqref{eq:LPVSumupConstraints}\quad Eq.~\eqref{eq:lpvmpcx0}, \quad \bs{X}\in\mathcal{X}, \quad \bs{u}\in \mathcal{U}. 
\end{align}

Algorithm~\ref{alg:LPVMPC} shows the process for obtaining the decision variable $\bs{X}$. The reference state is obtained from Algorithm~\ref{alg:local_reference} (line 1). With this reference state, we can compute $\bs{L}_N$ as described in Section~\ref{sec:Dyanmics} (line 5). Finally, the decision variable $\bs{X}$ is determined by solving Eq.~\eqref{eq:LPVMPC_experimental_formulation} (line 6).

\begin{algorithm}[tbp]
\caption{LPVMPC}
\label{alg:LPVMPC}
\DontPrintSemicolon
\SetAlgoVlined
\SetKwInOut{Input}{Input}
\SetKwInOut{Output}{Output}
\SetKwIF{IfNot}{ElseIfNot}{}{if not}{then}{else if not}{}{}
\Input{$\bs{z}_0$, $\Gamma_\text{r}$, $N$}
\Output{$\bs{X}$}
\nl $[x^\text{ref},y^\text{ref},\theta^\text{ref},v_x^\text{ref}] \gets $ Algorithm~\ref{alg:local_reference}\;
\nl $\{\bs{A}\} \gets \varnothing$, $\{\bs{B}\} \gets \varnothing$\;
\For{$ k = 1$ to $N+1$}{
   \nl \lIfNot{initialized}{
        $\bs{p}_k \gets [v_{x_k}^\text{ref},0,0,\theta_k^\text{ref}]$}
    \lElse{
        $\bs{p}_k\gets \bs{p}_{i-1|k+1}^\star$}
    \nl $\{\bs{A}\} \gets \bs{A}_k \cup \{\bs{A}\}$, $\{\bs{B}\} \gets \bs{B}_k \cup \{\bs{B}\}$\;
}
    \nl $\bs{L}_N \gets$ Eq.~\eqref{eq:Ln}\;
    \nl $\bs{X} \gets$ Eq.~\eqref{eq:LPVMPC_experimental_formulation}
\end{algorithm}

\subsection{DRCC-LPVMPC Experiment Setup}\label{section:drcclpvmpc setup}
By running the experiments in Section~\ref{section:LPVMPC} under the assumption of no disturbance, we can calculate the model discrepancy between the nonlinear dynamic model (described in Section~\ref{section:nonlinearModelsamples}) and the quasi-LPV model at each iteration. For an optimized output $\bs{X}_i$ obtained from Algorithm~\ref{alg:LPVMPC}, let $\bs{z}_i^{d}$ denote the nonlinear state using the first control input in $\bs{X}_i$, and $\bs{z}_i^l$ the corresponding quasi-LPV state. The sampled model difference can then be obtained by:
$\bs{\hat{\tau}}_{i} = \bs{z}_{i+1}^{d} - \bs{z}_{i+1}^l, i \in \mathbb{Z}_{\geq 0}.$
We further introduce random disturbances as additive uncertainty to the system, which contribute to the accumulation of model uncertainty:
\begin{equation}\label{eq:getxidata}
\hat{\bs{\xi}}_{i} = \bs{\hat{\tau}}_{i}+\mathbb{U}(-0.01,0.01),i\in\mathbb{Z}_{>0}.
\end{equation}
The error $\mathbb{U}(-0.01, 0.01)$ represents a uniformly distributed disturbance with values ranging from $-0.01$ to $0.01$. The purpose of this term is to simulate other uncertainties except the quasi-LPV process, the simulated $\bs{\hat{\xi}}$ shows how the used quasi-LPV model differs with the actual vehicle model due to any reasons, which is close to the real world condition as the only change in real world reflects on the model is a different sampled $\bs{\hat{\xi}}$. 

From the collected disturbance set, we randomly select $J$ samples to form the initial disturbance set $\bs{\bs{\Xi}} = \{\hat{\bs{\xi}}_0, \hat{\bs{\xi}}_1, \dots, \hat{\bs{\xi}}_J\}$, which is then used in the DRCC-LPVMPC. As the vehicle progresses, a new disturbance $\hat{\bs{\xi}}_i$ is calculated at each iteration and replaces the first element in the set, updating it as $\bs{\Xi} = \hat{\bs{\xi}}_i \cup \bs{\Xi}[1:]$.

Now we can convert Eq.~\eqref{eq:LPVMPC_experimental_formulation} into a DRCC problem, in this work, we will keep the hard constraints of the corridor and focus on the obstacle avoidance, which is converted to following based on Eq.~\eqref{eq:DRCCP11}. For all $j \in [J]$ and $k\in[N]$,
\small
\begin{equation}
\label{eq:DRCC-LPVMPC_formulation}
\begin{aligned}
        & \sigma - \epsilon \gamma_k \leq\frac{1}{N}\sum_{j\in[J]}q_j^k,\quad q_j^k+\gamma_k\leq s_j^k, \\
        & s_j^k\leq \bs{d}^o_k\bs{D}_k\bs{L}_N\bs{X}-c^o_k-\hat{w}_j^k+M_j(1-r_j^k), \\
        & s_j^k\leq M_jr_j^k, \quad \gamma_k \geq 0,\quad q_j^k\leq0,\quad s_j^k\geq0,\quad r_j^k\in\{0,1\}.
\end{aligned}
\end{equation}
\normalsize
Thus, the DRCC-LPVMPC formulation in Eq.~\eqref{DRCCLPVMPC_formulation} can be written as follows:
\begin{align}\label{DRCCLPVMPC_formulation_experiments}
    \min_{\bs{X}} \sum_{k=1}^N \{&P_k[(\bs{e}_0^6\bs{D}_k\bs{L}_N\bs{X}-x^\text{ref}_k)^2 + (\bs{e}_1^6\bs{D}_k\bs{L}_N\bs{X} - y^\text{ref}_k)^2] \nonumber\\
    &+ \rho_k\sum_{j\in[J]} (q_j^k+\gamma_k)\} \\
    \text{s.t.}&\quad Eq.~\eqref{eq:lpvmpcx0}, \quad Eq.~\eqref{eq:DRCC-LPVMPC_formulation},\quad \forall j\in [J], \quad \forall k\in[N], \nonumber\\
    & \quad Eq.~\eqref{eq:LPVSumupConstraints1},\quad  Eq.~\eqref{eq:LPVSumupConstraints2}, \quad \bs{X}\in\mathcal{X}, \quad \bs{u}\in \mathcal{U}.\nonumber
\end{align}

Algorithm~\ref{alg:DRCC-LPVMPC} outlines the process for calculating the decision variable in DRCC-LPVMPC. It follows a similar procedure to Algorithm~\ref{alg:LPVMPC} but includes an additional step for computing $\bs{H}_N$ (line 1-5). The decision variable \bs{X} is then obtained by solving Eq.~\eqref{DRCCLPVMPC_formulation_experiments} (line 6). Finally, the disturbance set is updated for using in the next iteration (line 7).
\begin{algorithm}[htbp]
\caption{DRCC-LPVMPC}
\label{alg:DRCC-LPVMPC}
\DontPrintSemicolon
\SetAlgoVlined
\SetKwInOut{Input}{Input}
\SetKwInOut{Output}{Output}
\SetKwIF{IfNot}{ElseIfNot}{}{if not}{then}{else if not}{}{}
\Input{$\bs{z}_0$, $R_\text{r}$, $N$}
\Output{$\bs{X}$}
\nl $[x^\text{ref},y^\text{ref},\theta^\text{ref},v_x^\text{ref}] \gets $ Algorithm~\ref{alg:local_reference}\; 
\nl $\{\bs{A}\} \gets \varnothing$, $\{\bs{B}\} \gets \varnothing$\;
\For{$ k = 1$ to $N+1$}{
   \nl \lIfNot{initialized}{
        $\bs{p}_k \gets [v_{x_k}^\text{ref},0,0,\theta_k^\text{ref}]$}
    \lElse{ $\bs{p}_k\gets \bs{p}_{i-1|k+1}^\star$}
    \nl $\{\bs{A}\} \gets \bs{A}_k \cup \{\bs{A}\}$, $\{\bs{B}\} \gets \bs{B}_k \cup \{\bs{B}\}$
}
    \nl $\bs{L}_N \gets$ Eq.~\eqref{eq:Ln}, $\bs{H}_N \gets$ Eq.~\eqref{eq:Hn}\;
    \nl $\bs{X} \gets$ Eq.~\eqref{DRCCLPVMPC_formulation_experiments}\;
    \nl $\hat{\bs{\xi}}_i \gets$ Eq.~\eqref{eq:getxidata}, $\bs{\Xi} = \hat{\bs{\xi}}_i \cup \bs{\Xi}[1:]$\;
\end{algorithm}
\vspace{-3mm}
\begin{figure}[h!]
	\begin{center}
		\includegraphics[width=3in]{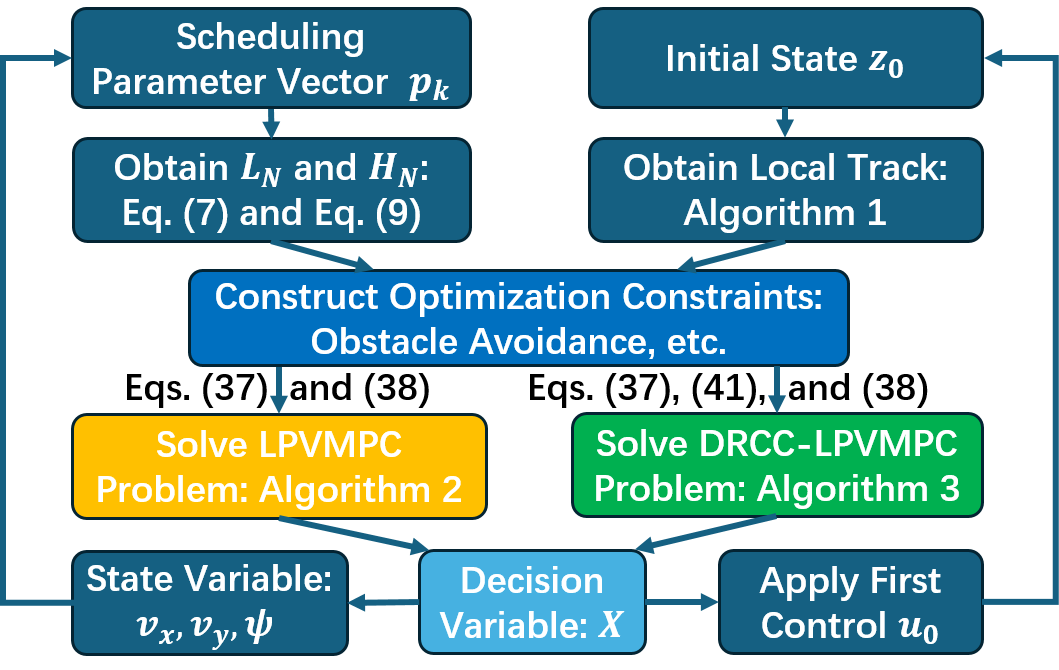}
           \vspace{-3mm}
		\caption{Overview of obtaining the control variable $\bs{X}$ from LPVMPC and DRCC-LPVMPC.}\label{fig:Complement Fig1}
	\end{center}
\end{figure}

In particular, both LPVMPC and DRCC-LPVMPC are used to compute the decision variable $\bs{X}$, which combines the initial state $z_0$ and the control inputs $u_k, k \in [N]$. The predicted states $[v_{x_k},v_{y_k},\psi_k]$ are then obtained using Eq.~\eqref{eq:lpv_rewrite}. Together with the control inputs $\delta_k$, these form the scheduling parameter vector $p_k = [v_{x_k},v_{y_k},\delta_k,\psi_k]$, which is passed to the next time step. Fig.~\ref{fig:Complement Fig1} summarizes the overall steps of LPVMPC and DRCC-LPVMPC.

\subsection{NMPC Setup}\label{section:nmpc setup}
Beyond LPVMPC and DRCC-LPVMPC, we also conducted NMPC experiments, following the formulation in~\cite{jain2020bayesrace}. Since~\cite{jain2020bayesrace} only addresses trajectory tracking without collision avoidance, we extend it by incorporating the same first-order obstacle avoidance constraints as in Eq.~\eqref{eq:LPVMPC_formulation}. 
The NMPCs are solved using the Interior Point Optimizer. The full formulation of the NMPC tracking problem is detailed in~\cite{jain2020bayesrace} and is omitted here for brevity.

\section{Numerical Results and Discussions}
\label{sec:result}

\subsection{LPVMPC}
We first present the experimental results of running the LPVMPC as described in Section~\ref{section:LPVMPC}. Fig.~\ref{fig:track_lpvmpc} shows the tracking results using LPVMPC, under the assumption that no uncertainties are considered. The distance to the obstacles is shown in Table~\ref{tab:Dis_compare_rho}. While performing obstacle avoidance, the maximum distance to the obstacles is only 0.00374 m. This means that if the disturbances exceed this threshold, the controller will direct the vehicle to an infeasible region of the problem. To simulate accumulated uncertainties, random disturbances from Table~\ref{tab:AppliedDisturbance} are added to the current state. We conducted 40 trials, and in all cases the vehicle collided with the first obstacle once additive noise pushed it into the safety region after approaching too closely. In rare instances, the vehicle managed to bypass the first two obstacles with very narrow margins, but even in these best-case scenarios it ultimately entered the safety region of the third obstacle and collided. One such case is illustrated in Fig.~\ref{fig:track_lpvmpc_uncertain}, where the vehicle successfully skirts the first two obstacles but is eventually driven inside the third obstacle’s safety region due to random disturbances.

\begin{figure}[!htb]
	\begin{center}
		\includegraphics[width=3in]{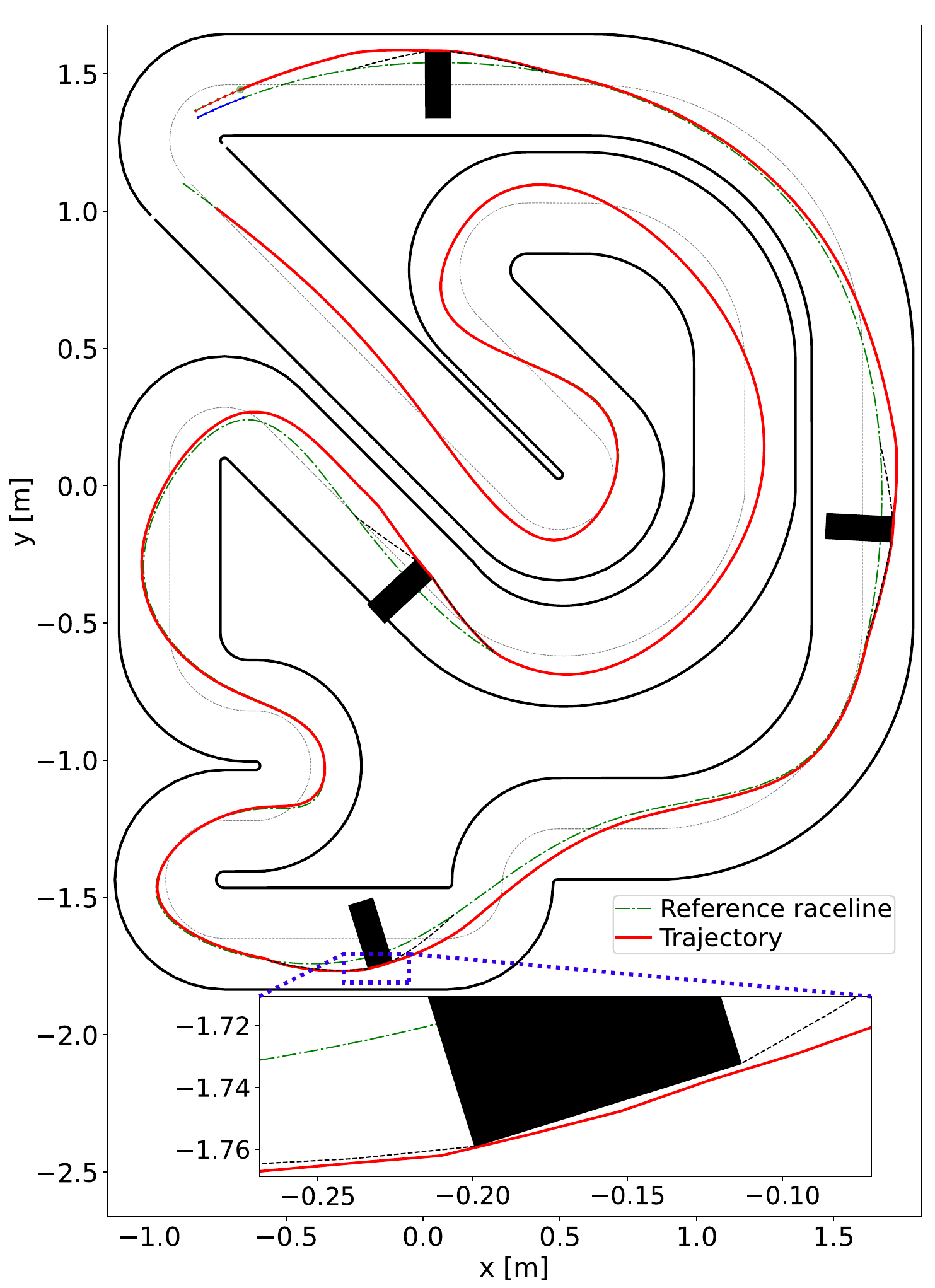}
        \vspace{-2mm}
		\caption{Tracking results using LPVMPC without uncertainties, showing the close distance between the path and obstacles. }\label{fig:track_lpvmpc}
	\end{center}
\end{figure}

\begin{table}[!thb]
	\centering
	\caption{NMPC, LPVMPC, and DRCC-LPVMPC Distance Statistics for Obstacles ($\epsilon = 0.10$, $\sigma = 0.08$)}
	\vspace{-2mm}
    \renewcommand{\arraystretch}{0.95}
	\begin{tabular}{lccccc}
		\toprule
		\textbf{Methods} & \textbf{$\rho_k$} & \textbf{Obs.} & \textbf{Mean Dist.} & \textbf{Max Dist.} & \textbf{Min Dist.} \\
		\midrule
            \multirow{4}{*}{NMPC} & - & 0 & 0.00132 & 0.00451 & 0.00039 \\
		& - & 1 & 0.00267 & 0.00289 & 0.00198 \\
		& - & 2 & 0.00148 & 0.00486 & 0.00253 \\
		& - & 3 & 0.00301 & 0.00415 & 0.00259 \\
		\midrule
		\multirow{4}{*}{LPVMPC} & - & 0 & 0.00197 & 0.00351 & 0.00056 \\
		& - & 1 & 0.00220 & 0.00258 & 0.00173 \\
		& - & 2 & 0.00327 & 0.00374 & 0.00262 \\
		& - & 3 & 0.00270 & 0.00318 & 0.00202 \\
		\midrule
		\multirow{19}{*}{DRCC-} & \multirow{4}{*}{-0.12} & 0 & 0.01911 & 0.02024 & 0.01738 \\
		& & 1 & 0.01488 & 0.01587 & 0.01371 \\
		& & 2 & 0.02339 & 0.02438 & 0.02222 \\
		& & 3 & 0.02094 & 0.02364 & 0.01880 \\
		\cmidrule{2-6}
		\multirow{12}{*}{LPVMPC} & \multirow{4}{*}{-0.14} & 0 & 0.02319 & 0.02449 & 0.02130 \\
		& & 1 & 0.01838 & 0.01947 & 0.01714 \\
		& & 2 & 0.02672 & 0.02783 & 0.02544 \\
		& & 3 & 0.02377 & 0.02546 & 0.02174 \\
		\cmidrule{2-6}
		& \multirow{4}{*}{-0.16} & 0 & 0.02718 & 0.02875 & 0.02503 \\
		& & 1 & 0.02120 & 0.02221 & 0.01997 \\
		& & 2 & 0.02964 & 0.03088 & 0.02822 \\
		& & 3 & 0.02633 & 0.02828 & 0.02402 \\
		\cmidrule{2-6}
		& \multirow{4}{*}{-0.18} & 0 & 0.03084 & 0.03249 & 0.02859 \\
		& & 1 & 0.02429 & 0.02544 & 0.02293 \\
		& & 2 & 0.03321 & 0.03457 & 0.03169 \\
		& & 3 & 0.03264 & 0.03401 & 0.03113 \\
		\cmidrule{2-6}
		& \multirow{4}{*}{-0.20} & 0 & 0.03506 & 0.03700 & 0.03261 \\
		& & 1 & 0.02730 & 0.02879 & 0.02561 \\
		& & 2 & 0.03010 & 0.03183 & 0.02822 \\
		& & 3 & 0.03533 & 0.03683 & 0.03368 \\
		\bottomrule
	\end{tabular}
	\label{tab:Dis_compare_rho}
\end{table}

\begin{table}[htbp]
	\caption{Ranges of additive disturbances}
	\vspace{-2mm}
        \hspace{13mm}
	\label{tab:AppliedDisturbance}
	\renewcommand{\arraystretch}{1.1} 
	\begin{tabularx}{0.8\columnwidth}{|c |c|c |}
		\hline
		\textbf{State Variable} & \textbf{Min} & \textbf{Max}  \\ \hline
		$x$ (m)          & $-0.5\times10^{-2}$          & $0.5\times10^{-2}$            \\ \hline
		$y$ (m)        & $-0.5\times10^{-2}$         & $0.5\times10^{-2}$        \\ \hline
		$\varphi$ (rad)         & $-0.5\times10^{-4}$          & $0.5\times10^{-4}$        \\ \hline
		$v_x$ (m/s)          & $-0.1\times10^{-4}$        & $0.1\times10^{-4}$          \\ \hline
		$v_y$ (m/s)         & $-0.1\times10^{-5}$         & $0.1\times10^{-4}$        \\ \hline
		$\omega$ (rad/s)      & $-0.1\times10^{-4}$         & $0.1\times10^{-4}$         \\ \hline
	\end{tabularx}
\end{table}

\begin{figure}[ht]
	\begin{center}
		\includegraphics[width=3.1in]{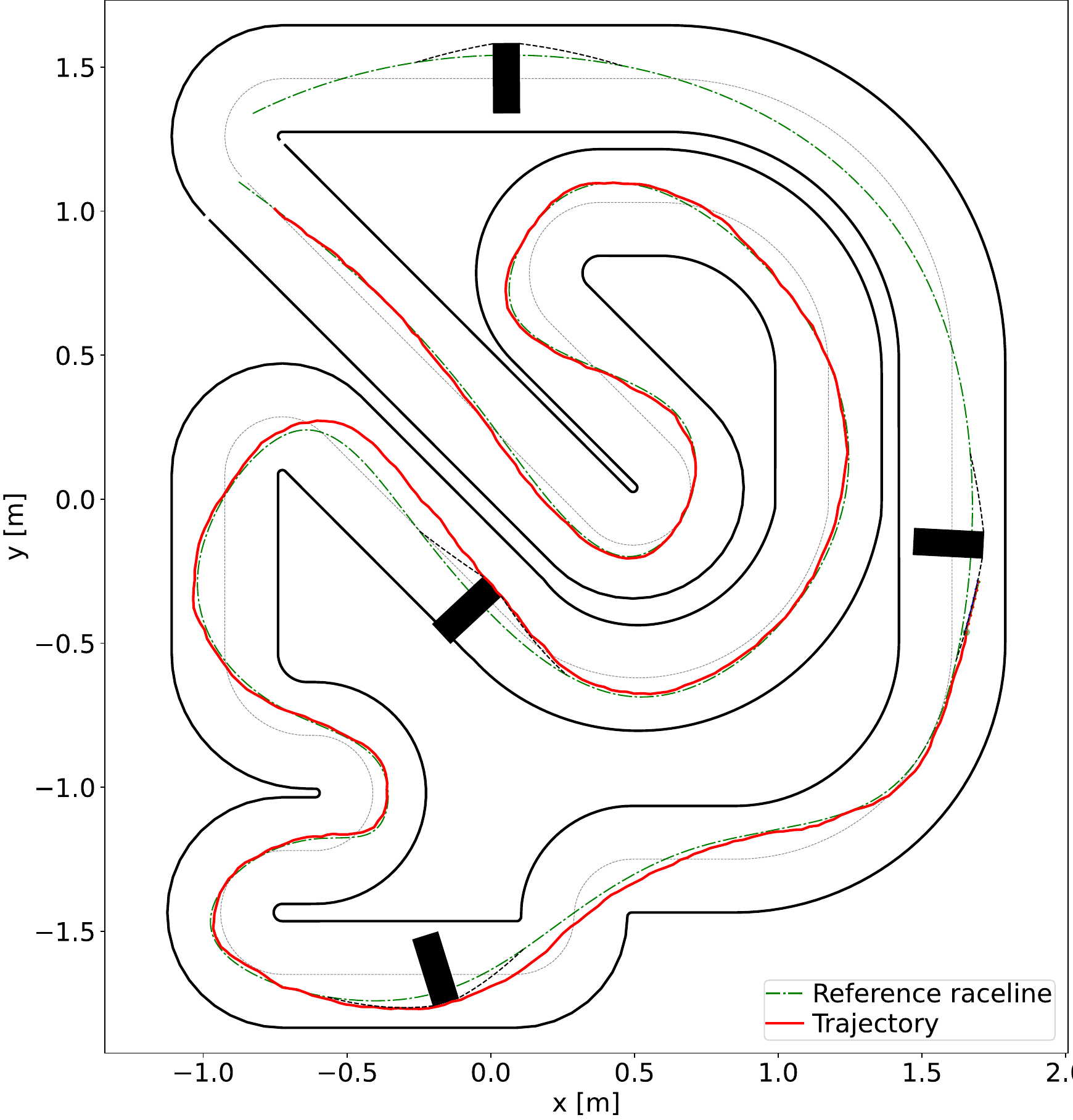}
           \vspace{-3mm}
		\caption{The best-case scenario tracking results using LPVMPC with random disturbances added to the state at each iteration, where the solver fails at the third obstacle due to an infeasible problem.}\label{fig:track_lpvmpc_uncertain}
	\end{center}
        \vspace{-5mm}
\end{figure}
\begin{figure}[t]
	\begin{center}
		\includegraphics[width=3.in]{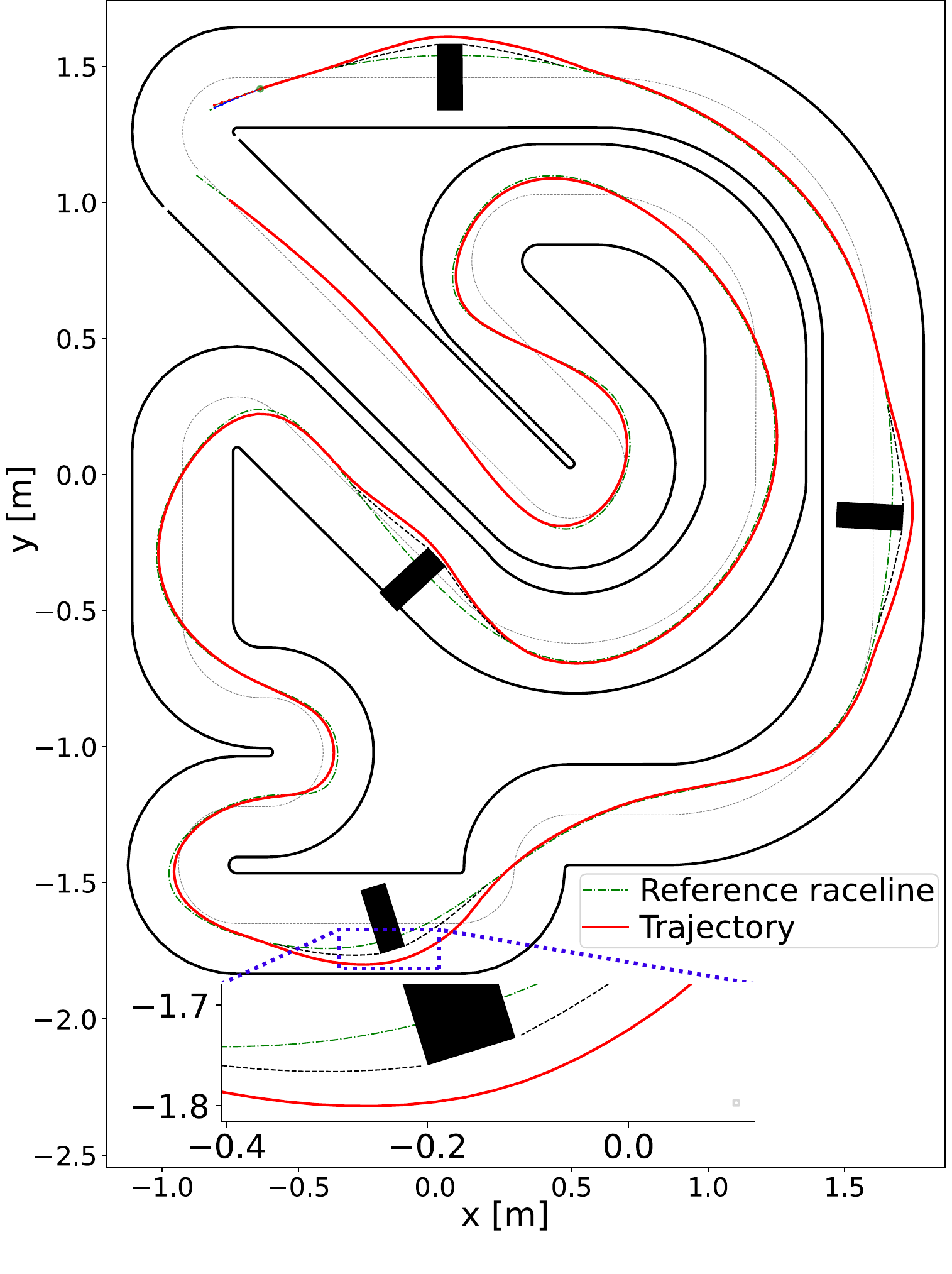}
           \vspace{-3mm}
		\caption{Tracking results using DRCC-LPVMPC without uncertainties under $\rho_k = -0.18$. The inset figure shows the zoomed-in tracking trajectory versus the reference, where the distance between the path and obstacles is larger.}\label{fig:track_drcclpvmpc}
	\end{center}
    \vspace{-5mm}
\end{figure}
\begin{figure}[ht!]
	\begin{center}
		\includegraphics[width=3.1in]{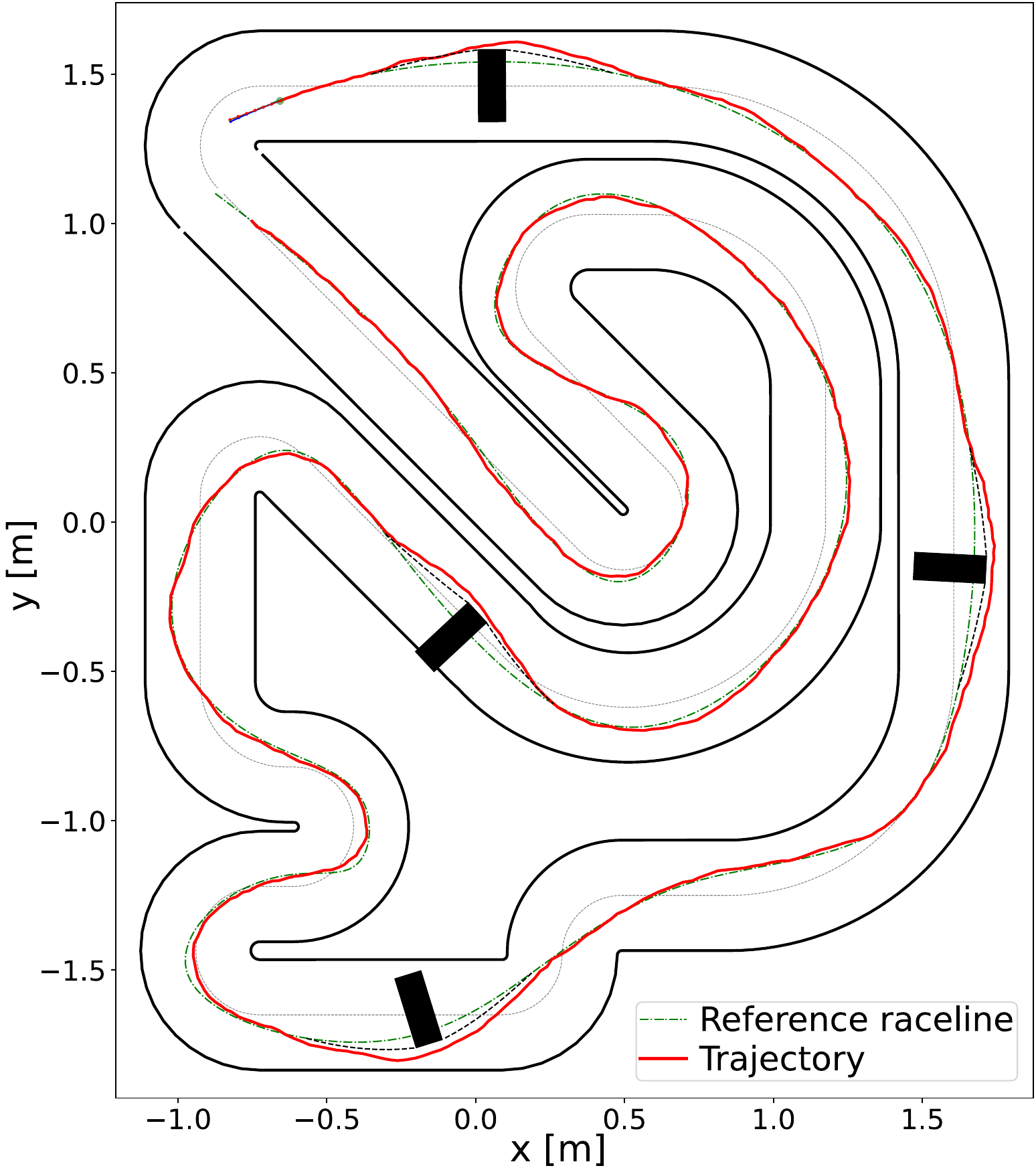}
           \vspace{-3mm}
		\caption{ Tracking results using DRCC-LPVMPC with $\rho_k = -0.20$, where random disturbances are added to the state at each iteration, illustrating a safer driving trajectory under uncertainties.}\label{fig:track_drcclpvmpc_uncertain}
	\end{center}
    \vspace{-7mm}
\end{figure}

We note that different works represent obstacles in various ways (e.g., circles, tubes, or grid-based methods), and by modifying the cost or constraint formulations, linear or nonlinear MPC schemes can achieve successful avoidance. In this study, however, we intentionally adopt a consistent obstacle representation and avoidance scheme for both LPVMPC and our proposed DRCC-LPVMPC, so that the benefits of distributional robustness can be highlighted through a fair comparison.

\subsection{DRCC-LPVMPC}
A common way to address collision is to add an ``inflation radius" to obstacles, effectively ``expanding" them with a safety gap, similar to the concepts used in simultaneous localization and mapping. The original purpose of adding the inflation radius is to treat the vehicle as a point by adding the maximum radius of the vehicle body to the obstacle boundary. By further increasing the inflation radius, the vehicle may hit the inflated boundary, but it is still considered safe relative to the actual obstacle boundary. However, increasing the inflation radius reduces the feasible region for the controller and does not change the fact that the constraints remain hard constraints. If uncertainty causes the vehicle to enter an infeasible region, such as the safety region shown in Fig.~\ref{fig:feasible_region}, the controller will fail to find a solution, regardless of the inflation radius. By converting the hard constraints into chance constraints in the DRCC-LPVMPC, we not only expand the feasible region (increasing robustness) but also enable the vehicle to bypass obstacles at a larger distance (improving safety). This is another key contribution of this work.

Numerical experiments for DRCC-LPVMPC were conducted as described in Section~\ref{section:drcclpvmpc setup}. Similar to LPVMPC, the first tests were run without random additive disturbances to the current state. Fig.~\ref{fig:track_drcclpvmpc} shows the tracking performance with $\rho_k = -0.18$, $\epsilon = 0.10$, and $\sigma = 0.08$. Compared to Fig.~\ref{fig:track_lpvmpc}, it is clear that DRCC-LPVMPC allows the vehicle to avoid obstacles with larger clearance, yielding safer trajectories even under uncertainty. Fig.~\ref{fig:track_drcclpvmpc_uncertain} presents the tracking performance of DRCC-LPVMPC with $\rho_k = -0.20$, $\epsilon = 0.10$, and $\sigma = 0.08$, while introducing random additive disturbances to the current state. Unlike LPVMPC, DRCC-LPVMPC successfully avoids the obstacles and maintains a safe distance from them, even in the presence of disturbances. Readers can refer to the companion video clip for a comparison of tracking performance.

We also present the mean, maximum, and minimum distances between the vehicle and each obstacle during the avoidance process, as shown in Table~\ref{tab:Dis_compare_rho}. It is evident that as the absolute value of $\rho_k$ increases, the planner increases the distance between the vehicle and the obstacles to enhance robustness. Additionally, we show the best lower bound value found by the DRCC-LPVMPC in Table~\ref{tab:LB_compare_epsilon}. As the risk tolerance $\epsilon$ decreases from 0.30 to 0.05, the planner adjusts to increase the lower bound of the chance constraints, further improving robustness.

\begin{table}[!htb]	
	\centering
	\caption{DRCC-LPVMPC lower bound (LB) statistics for obstacles ($\sigma = 0.08, \rho_k = -0.16$)}
	\vspace{-2mm}
    \renewcommand{\arraystretch}{0.93}
	\begin{tabular}{lccccc}
		\toprule
		\textbf{Method} & \textbf{$\epsilon$} & \textbf{Obs.} & \textbf{Mean LB} & \textbf{Max LB} & \textbf{Min LB} \\
		\midrule
		\multirow{19}{*}{DRCC-} & \multirow{4}{*}{0.05} & 0 & 500.01312 & 500.01390 & 500.01203 \\
		& & 1 & 500.01112 & 500.01162 & 500.01050 \\
		& & 2 & 475.09651 & 475.09722 & 475.09563 \\
		& & 3 & 487.55336 & 500.01436 & 475.09230 \\
		\cmidrule{2-6}
		\multirow{12}{*}{LPVMPC} & \multirow{4}{*}{0.10} & 0 & 466.73778 & 500.01391 & 450.09370 \\
		& & 1 & 500.01108 & 500.01159 & 500.01047 \\
		& & 2 & 466.73512 & 500.01408 & 450.09539 \\
		& & 3 & 475.05289 & 500.01353 & 450.09199 \\
		\cmidrule{2-6}
		& \multirow{4}{*}{0.15} & 0 & 450.06719 & 500.01384 & 425.09329 \\
		& & 1 & 475.04104 & 500.01121 & 425.10142 \\
		& & 2 & 450.07035 & 500.01424 & 425.09401 \\
		& & 3 & 462.55265 & 500.01353 & 425.09171 \\
		\cmidrule{2-6}
		& \multirow{4}{*}{0.20} & 0 & 433.40256 & 500.01384 & 400.09359 \\
		& & 1 & 433.40379 & 500.01128 & 400.09962 \\
		& & 2 & 466.70704 & 500.01430 & 400.09323 \\
		& & 3 & 450.05220 & 500.01439 & 400.09135 \\
		\cmidrule{2-6}
		& \multirow{4}{*}{0.30} & 0 & 350.09387 & 350.09774 & 350.09162 \\
		& & 1 & 450.03601 & 500.01056 & 350.08816 \\
		& & 2 & 350.09133 & 350.09156 & 350.09091 \\
		& & 3 & 350.08993 & 350.09035 & 350.08921 \\
		\bottomrule
	\end{tabular}
	\label{tab:LB_compare_epsilon}
\end{table}

\subsection{NMPC}
As shown in Table~\ref{tab:Dis_compare_rho}, NMPC with hard constraints fails to ensure safe and robust obstacle avoidance. Table~\ref{tab:solution time} further shows that NMPC’s computation time increases sharply with obstacle-avoidance constraints, leading to frequent solution failures. In contrast, LPVMPC and DRCC-LPVMPC achieve real-time feasibility, with the linearized model greatly reducing computational burden. While DRCC-LPVMPC requires slightly more computation than LPVMPC due to uncertainty handling, its runtime remains well within real-time limits. Overall, DRCC-LPVMPC provides a practical balance between robustness and efficiency.

\begin{table}[!t]
	\centering
	\caption{Average computation time over 30 different trials under scenarios without and with obstacle avoidance}
	\vspace{-2mm}
	\begin{tabularx}{\columnwidth}{|l|c|c|}
		\hline
		\textbf{Methods} & No obstacle avoidance &  With obstacle avoidance  \\
		\hline
		NMPC & 0.026 s & 0.37 s\\
		\hline
		LPVMPC & 0.0032 s& 0.0047 s\\
		\hline
		DRCC-LPVMPC & 0.0032 s& 0.0098 s\\
		\hline
    \end{tabularx}
	\label{tab:solution time}
\end{table}

\begin{figure}[t!]
    \centering
\includegraphics[width=\linewidth]{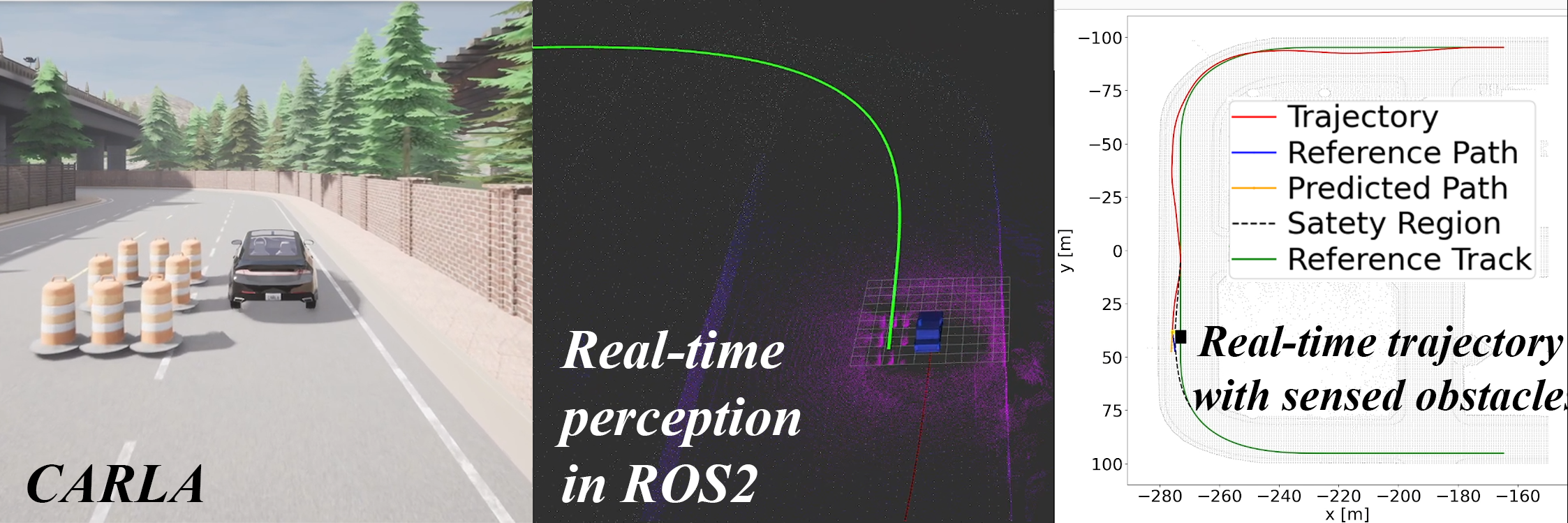}
   \vspace{-6mm}
    \caption{High-fidelity CARLA simulation of obstacle avoidance using DRCC-LPVMPC. A full video demonstration and package are
available on GitHub.}
    \label{fig:experimentCarla}
\end{figure}

\begin{figure*}[t!]
    \hspace{-2.5mm}
\subfigure[]{\includegraphics[height = 1.8 in]{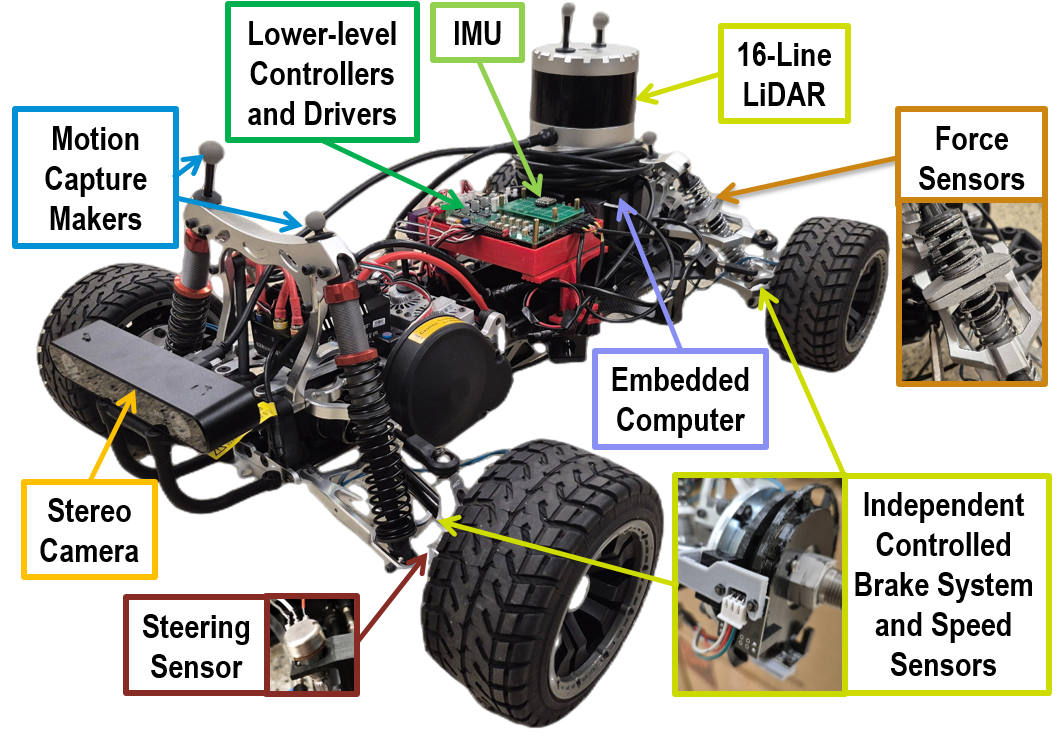}
	\label{fig:experimentVehicle}}
    \hspace{-3.5mm}
\subfigure[]{
\includegraphics[height = 1.8 in]{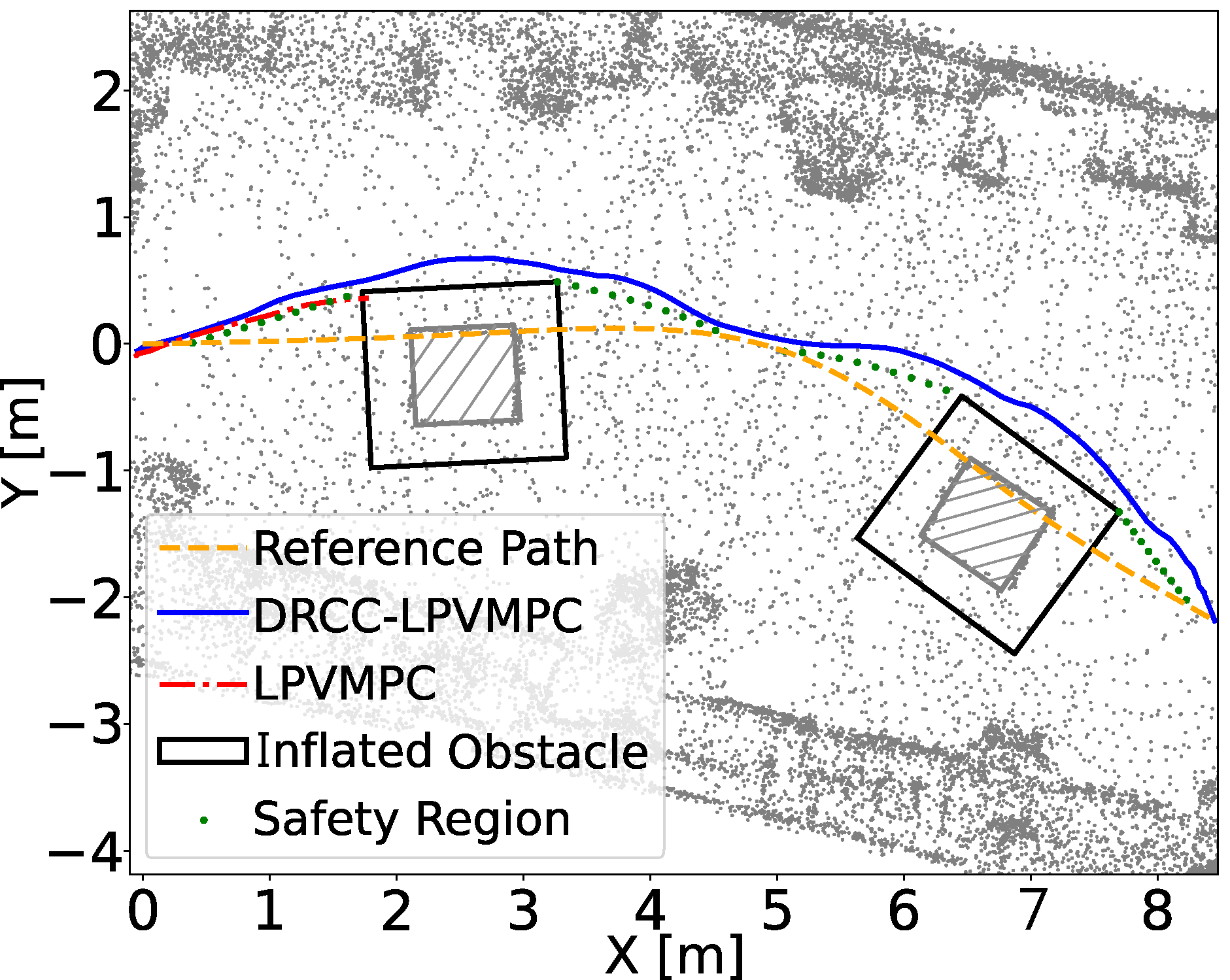}
	    \label{fig:realExpTracking}}
    \hspace{-3mm}
\subfigure[]{
\includegraphics[height = 1.8 in]{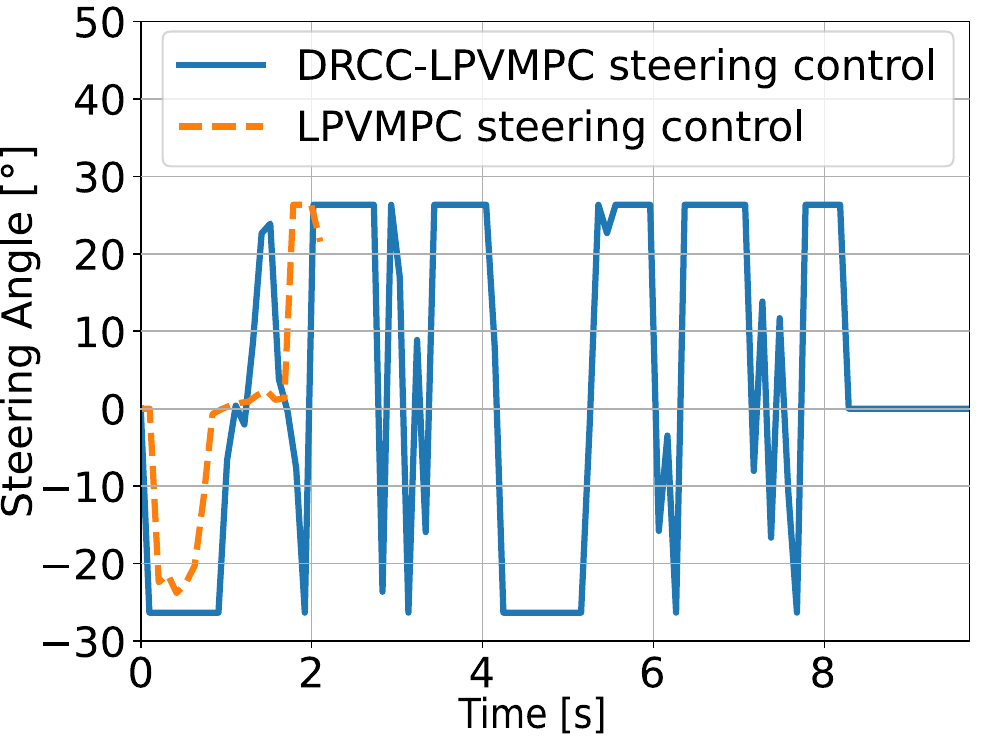}
	 \label{fig:realExpSteering}}
\vspace{-5mm}
  \caption{(a) 1:5 scaled experimental vehicle platform. $(x, y, \psi)$ is estimated using LiDAR-based 3D SLAM, while $v_x, v_y,$ and $\omega$ are obtained from odometry and an onboard IMU. An NVIDIA Jetson AGX Orin running ROS serves as the embedded computer, providing real-time control through steering and duty cycle/velocity commands. (b) Tracking performance of the scaled vehicle under LPVMPC (red) and DRCC-LPVMPC (blue). The obstacle is represented by the centered gray box with its inflated region shown in black. The yellow line denotes the reference trajectory, and the green line indicates the upper bound of the safety region. DRCC-LPVMPC maintains a larger clearance, ensuring safe avoidance, while LPVMPC collides with the obstacle. (c) Steering control inputs generated by LPVMPC and DRCC-LPVMPC during the real-world experiment. See the companion video for the real-world tracking performance.}
\vspace{-3mm}
\end{figure*}

\subsection{CARLA Simulation}
In addition to the numerical experiments, we conducted high-fidelity simulation studies in CARLA~\cite{Dosovitskiy17}, a widely used platform for sim-to-real autonomous driving research. The vehicle’s tire model was obtained using the method in~\cite{fang2025fine,FangMECC2024}. Fig.~\ref{fig:experimentCarla} shows the CARLA simulation of obstacle avoidance using DRCC-LPVMPC. Unknown obstacles are detected online through sensors and incorporated into the controller in real time. The vehicle model in CARLA uses an unknown nonlinear model, while the controller relies on a bicycle model for estimation. DRCC-LPVMPC successfully avoids three obstacles with safe margins, demonstrating robustness under model mismatch. A full video demonstration and ROS2 package are available on \href{https://github.com/Binghamton-ACSR-Lab/DRCCLPVMPC}{GitHub}.

For clarity, we report in the paper only the results with static obstacles, while additional experiments with dynamically sensed obstacles are provided in the GitHub repository. These results show that moving obstacles impose greater challenges on motion control and solution feasibility, as the vehicle must continuously adapt to dynamic constraints. Nevertheless, our proposed DRCC-LPVMPC controller consistently demonstrated robustness, safely avoiding all obstacles with sufficient clearance. In contrast, LPVMPC frequently led the vehicle too close to the safety region, resulting in collisions.

\section{Real World Experiment}\label{sec:RealExp}

To further validate our approach, we performed real-world experiments using a 1:5 scaled vehicle, shown in Fig.~\ref{fig:experimentVehicle}. The test scenario required the vehicle to bypass box-shaped obstacles. Both LPVMPC and DRCC-LPVMPC were implemented with identical controller configurations (tracking weights, vehicle dynamics, and maximum vehicle speed of $1.5\,\text{m/s}$). The tracking results in Fig.~\ref{fig:realExpTracking} show that LPVMPC (red trajectory) collided with the obstacle during the turning maneuver, whereas DRCC-LPVMPC successfully avoided the obstacle and reached the target. 

Moreover, DRCC-LPVMPC maintained a larger safety margin, as reflected in Fig.~\ref{fig:realExpSteering}, where the steering behavior leads to trajectories that remain farther from obstacles while still preserving accurate path tracking. This robustness arises from the data-driven disturbance samples $\hat{w}_j^k$, computed online as discrepancies between quasi-LPV predictions and sensor measurements. Although not detailed in this paper, these samples directly shaped the ambiguity set used in the DRCC formulation.

Finally, Fig.~\ref{fig:realExpTracking} highlights that, while LPVMPC still produced mathematically feasible trajectories, the real vehicle collided with the obstacle due to localization errors and actuation disturbances. This outcome is consistent with our simulation findings, underscoring that controllers which drive too close to obstacles are inherently fragile under real-world uncertainty, whereas DRCC-LPVMPC preserves safety margins and robustness.



\section{Conclusion}
\label{sec:concl}
This paper presented DRCC-LPVMPC, a data-driven distributionally robust control framework for autonomous driving with emphasis on safety-critical tasks such as obstacle avoidance. By consolidating discrepancies between STM dynamics, full vehicle models, and sensing disturbances into an additive disturbance term, the original MPC problem was reformulated as a DRCC problem within a quasi-LPV representation. The resulting convex approximation can be solved online via QP using real-time disturbance samples. Simulations demonstrated superior feasibility and safety margins compared to LPVMPC and NMPC, while scaled-vehicle experiments confirmed robustness and improved obstacle avoidance. Remaining challenges include accurate scheduling parameter estimation in highly dynamic scenarios, motivating future work on refined linearization and adaptive disturbance modeling to further enhance robustness and generalization.

\appendix

\subsection{Proof of Theorem \ref{theorem:1}}
\label{app1}
To reformulate, according to Eq.~\eqref{eq:DRCCPs4}, 
\begin{align}\label{eq:DRCCPs5}
    \min_{\lambda_k \geq0}\{&\lambda_k\sigma - \frac{1}{N}\sum_{j\in[J]} \inf_{\tilde{w}}[\lambda_k\|\tilde{w}-\hat{w}_j\|- \nonumber\\
    & \mathbb{I}(\bs{d}_k\bs{D}_k\bs{L}_N\bs{X}-c_k<\tilde{w}^k,\exists k\in[N])]\} \leq \epsilon \nonumber
\end{align}
is equivalent to
\begin{align*}
    &\min_{\lambda_k\geq0}\{\lambda_k\sigma-\frac{1}{N}\sum_{j\in[J]}\nonumber\min\{\min_{k\in[N]}\\
    &\lambda_k\max\{\bs{d}_k\bs{D}_k\bs{L}_N\bs{X}-c_k-\hat{w}_j^k,0\}-1,0\}\} \leq \epsilon,
\end{align*}
by replacing with Eq.~\eqref{eq:DRCCPfact}. Then, we obtain
\begin{equation*}
    \min_{\lambda_k\geq0}\{\lambda_k\sigma-\frac{1}{N}\sum_{j\in[J]}\lambda_k f(\bs{X},\hat{w}_j)-1,0\}\leq\epsilon,
\end{equation*}
which can be expressed as
\begin{equation*}
    \lambda_k\sigma-\epsilon\leq\frac{1}{N}\sum_{j\in[J]}\min\{\lambda_k f(\bs{X},\hat{w}_j)-1,0\}.
\end{equation*}
This completes the proof.
\subsection{Proof of Theorem \ref{theorem:2}}
\label{app2}
To prove the theorem, we define $\mathbb{X}_D^1$ as the set in Eq.~\eqref{eq:DRCCP10}, then we only need to prove $\mathbb{X}_D^1\subseteq \mathbb{X}_D$ and $\mathbb{X}_D\subseteq \mathbb{X}_D^1$.
For $\mathbb{X}_D^1\subseteq \mathbb{X}_D$: Given $\bs{X}\in \mathbb{X}_D^1$, there exists $(\gamma_k,q,\bs{X})$ which satisfies Eq.~\eqref{eq:DRCCP10}. For each $k\in [N]$, Eq.~\eqref{eq:DRCCP10_1} and Eq.~\eqref{eq:DRCCP10_2} imply that
    \begin{subequations}
        \begin{align}
            q_j^k&\leq \min\{\min_{k\in [N]}\max\{\bs{d}_k\bs{D}_k\bs{L}_N\bs{X}-c_k-\hat{w}_j^k\}-\gamma_k,0\} \nonumber\\
            &=\min\{f(\bs{X},\hat{w}_j)-\gamma_k,0\}.\nonumber
        \end{align}
    \end{subequations}
    therefore, Eq.~\eqref{eq:DRCCP10_1} gives
    \[
    \sigma-\epsilon\gamma_k\leq\frac{1}{N}\sum_{j\in [J]}q_j^k\leq\frac{1}{N}\sum_{j\in [J]}\min\{f(\bs{X},\hat{w}_j)-\gamma_k,0\},
    \]
    which is the same as Eq.~\eqref{eq:DRCCP9_1}, thus $\bs{X}\in \mathbb{X}_D$.
    
For $\mathbb{X}_D\subseteq \mathbb{X}_D^1$: Given $\bs{X} \in \mathbb{X}_D$, there exists $(\gamma_k,\bs{X})$ that satisfies Eq.~\eqref{eq:DRCCP9}. Then, we can define
    \small
    \[
    q_j^k = \min_{k\in [N]}\{\max\{\bs{d}_k\bs{D}_k\bs{L}_N\bs{X}-c_k-\hat{w}_j^k,0\}-\gamma_k,0\},\forall j\in[J].
    \]
    \normalsize
    Next, Eq.~\eqref{eq:DRCCP9_1} becomes
    \small
    \begin{subequations}
        \begin{align}
            &\sigma-\epsilon\gamma_k  \nonumber \\
            &\leq\frac{1}{N}\sum_{j\in[J]}\min\{\min_{k\in [N]}\max\{\bs{d}_k\bs{D}_k\bs{L}_N\bs{X}-c_k-\hat{w}_j^k,0\}-\gamma_k,0\}\nonumber\\
            &=\frac{1}{N}\sum_{j\in[J]}q_j^k.\nonumber
        \end{align}
    \end{subequations}
    \normalsize
    and
    \begin{align}
        &q_j^k+\gamma_k = \min_{k\in [N]}\{\max\{\bs{d}_k\bs{D}_k\bs{L}_N\bs{X}-c_k-\hat{w}_j^k,0\}\}\nonumber\\
        &\leq \max\{\bs{d}_k\bs{D}_k\bs{L}_N\bs{X}-c_k-\hat{w}_j^k,0\},\forall j\in [J], \forall k \in [N].\nonumber
    \end{align}
    In addition, the expression of $q_j^k$ implies $q_j^k\leq0$,$\gamma_k\geq0$, thus $\bs{X}\subseteq \mathbb{X}_D^1$.
This completes the proof.

\subsection{Proof of Corollary \ref{corollary1}}
\label{app3}
Eq.~\eqref{eq:DRCCP10_2} is equivalent to
\[
q_j^k+\gamma_k \leq \max\{\min_{k\in[N]}\bs{d}_k\bs{D}_k\bs{L}_N\bs{X}-c_k-\hat{w}_j^k,0\},\forall j\in[J].
\]
The outer maximum on the right-hand side can be linearized by introducing a binary variable $g_j^k$, a continuous variable $s_j^k$, and a sufficiently large constant $M_j$. If $g_j^k=0$, then Eq.~\eqref{eq:DRCCP11_3} is relaxed (inactive) and $q_j^k+\gamma_k\leq s_j^k =0$, which covers the case $\bs{d}_k\bs{D}_k\bs{L}_N\bs{X}-c_k-\hat{w}_j^k \leq 0$ in Eq.~\eqref{eq:DRCCP10_2}. If $g_j^k=1$, then Eq.~\eqref{eq:DRCCP11_4} is relaxed (inactive) and $q_j^k+\gamma_k\leq s_j^k \leq \bs{d}_k\bs{D}_k\bs{L}_N\bs{X}-c_k-\hat{w}_j^k$, which covers the case $\bs{d}_k\bs{D}_k\bs{L}_N\bs{X}-c_k-\hat{w}_j^k \geq 0$ in Eq.~\eqref{eq:DRCCP10_2}. This completes the proof.

\subsection{Proof of Theorem \ref{theorem:RecursiveFeasibility}}
\label{app4}
 
        To formally prove recursive feasibility, we adopt binary initialization, a common MPC strategy~\cite{li2021distributionally,Hewing2020RobustAdaptive,Giselsson2014Realtime}. This approach ensures that if the controller becomes infeasible at the current step, it can initialize using the predicted state from the previous step, thereby maintaining feasibility. Specifically, while the controller is feasible at the current time $t$, it uses the latest state $\bs{z}_t$. If the problem becomes infeasible, the controller uses the predicted state $\bs{z}^*_{t|t-1}$ obtained from the optimization at the previous time step $t-1$.
        
        To prove recursive feasibility, we need to show that a feasible solution exists for the DRCC-LPVMPC at time $t+1$ if $\bar{\bs{z}}_{t+1} = \bar{\bs{z}}_{t+1|t}$. Specifically, assuming the DRCC-LPVMPC is feasible at time $t$ with the corresponding optimal solution $\bs{X}^*_t = \{\bs{z}_t,\bs{u}_t^*,\bs{u}_{t+1}^*,\dots,\bs{u}_{t+N-1}^*\}$, we need to prove that the solution $\bs{X}_{t+1} = \{\bs{z}_{t+1},\bs{u}_{t+1}^*,\bs{u}_{t+2}^*,\dots,\bs{u}_{t+N-1}^*,\bs{0}_{1\times2}\}$ is feasible at time $t+1$.
			
		Adopting binary initialization ensures that $\bs{X}_{t+1}$ satisfies the system constraints Eq.~\eqref{eq:systemConst} from time $t+1$ to $t+N-1$. Now, we need to prove $\bs{X}_{t+1}$ is feasible for the system constraints at $t+N$ and the terminal constraints are satisfied as well. To proceed, we first rewrite the system constraints as
			\begin{align}
				&\bs{d}_{t+k}\bs{z}_{t+k} - c_{t+k} - \hat{w}_j^{t+k} \nonumber \\
				&= \bs{d}_{t+k}(\bs{E}_t^{t+k}\bs{z}_t + \sum_{j=1}^k (\bs{E}_j^k\bs{B}_{j-1}\bs{u}_{j-1}^\top + \bs{E}_j^k\bs{\xi}_{j-1}^\top) + \bs{B}_k \bs{u}_k^\top) \nonumber \\
				&- c_{t+k}- \hat{w}_j^{t+k}. \label{eq:rewriteSysCons1}
			\end{align}
			Applying Eq.~\eqref{eq:rewriteSysCons1} to Corollary 1 and denoting the system with expectation, we obtain 
			\begin{align*}
				&\bs{d}_{t+k}\bar{\bs{z}}_{t+k}  = \bs{d}_{t+k}(\bs{E}_t^{t+k}\bar{\bs{z}}_t + \sum_{j=1}^k \bs{E}_j^k\bs{B}_{j-1}\bs{u}_{j-1}^\top + \bs{B}_k \bs{u}_k^\top) \nonumber \\
				&\leq M_j + c_{t+k}.
			\end{align*}
			Since $\bar{\bs{z}}_{t+N}$ satisfies the terminal constraints in Eq.~\eqref{Eq:systemBound}, at $t+N$, we have
			\begin{align*}
				&\bs{d}_{t+N}\bar{\bs{z}}_{t+N} \nonumber \\
				& = \bs{d}_{t+N}(\bs{E}_t^{t+N}\bar{\bs{z}}_t + \sum_{j=1}^{t+N} \bs{E}_j^{t+N}\bs{B}_{j-1}\bs{u}_{j-1}^\top + \bs{B}_{t+N} \bs{u}_{t+N}^\top) \nonumber \\
				&\leq M_j + c_{t+N}. \label{eq:TerminalCons3}
			\end{align*}
			Thus, the system constraints are satisfied for $\bs{X}_{t+1}$ at $t+N$.
			
			Finally, since $\bs{u}_{t+N} = \bs{0}_{1\times2}$ and $\bar{\bs{z}}_{t+N}\in\mathcal{Z}_T$, from Eq.~\eqref{eq:terminalConstraints} we obtain
			\begin{align*}
				&\bar{\bs{z}}_{t+N+1} = \bs{A}_{t+N}\bar{\bs{z}}_{t+N} + \bs{B}_{t+N}\bs{u}_{t+N} = \bs{A}_{t+N}\bar{\bs{z}}_{t+N} \in \mathcal{Z}_T.
			\end{align*}
			Therefore, the terminal constraints are satisfied for $\bs{X}_{t+1}$, completing the proof of recursive feasibility.

\ifCLASSOPTIONcaptionsoff
  \newpage
\fi





\bibliographystyle{IEEEtran}
\bibliography{IEEEabrv,Bibliography}
\begin{biography}[{\includegraphics[width=1in,height=1.25in,clip,keepaspectratio]{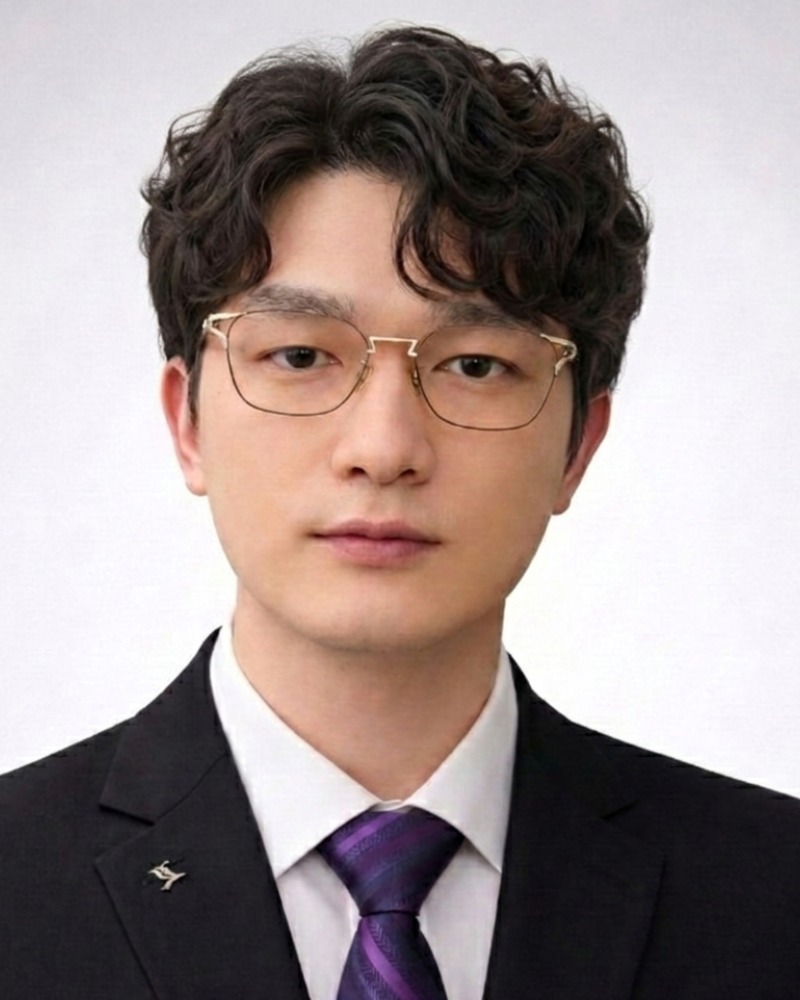}}]{Shiming Fang} received the B.S. degree in Mechanical Engineering from Wuhan University of Technology, Wuhan, China in 2016, and the M.S. degree in Mechanical Engineering from University of Birmingham, Birmingham, UK in 2017. He earned his Ph.D. in Mechanical Engineering at Binghamton University, Binghamton, NY, USA, in 2025. His research focuses on modeling, estimation, and deep reinforcement learning for robotics and autonomous robust control. He is currently a Postdoctoral Researcher at the State University of New York Polytechnic Institute, Utica, NY, USA, where he is developing deep reinforcement learning-based control for nanowires and energy-efficient spiking neural network models.
\end{biography}
\begin{biography}[{\includegraphics[width=1in,height=1.25in,clip,keepaspectratio]{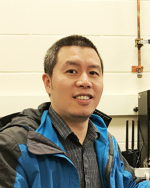}}]{Xilin Li} received the B.S. degree in process equipment and control engineering and the M.S. degree in chemical machinery from Zhejiang University, China, in 2007 and 2010, respectively. He earned his Ph.D. in Mechanical Engineering from Binghamton University, Binghamton, NY, USA, in 2025, specializing in advanced robotics with an emphasis on autonomous vehicles. His research primarily focuses on autonomous trajectory generation and optimization, as well as nonlinear model predictive trajectory control. He is currently a Postdoctoral Researcher at the Syracuse Center of Excellence at Syracuse University, Syracuse, NY, USA, where he is developing sensor networks and adaptive building technologies for environmental monitoring and energy efficiency.
\end{biography}
\begin{biography}[{\includegraphics[width=1in,height=1.25in,clip,keepaspectratio]{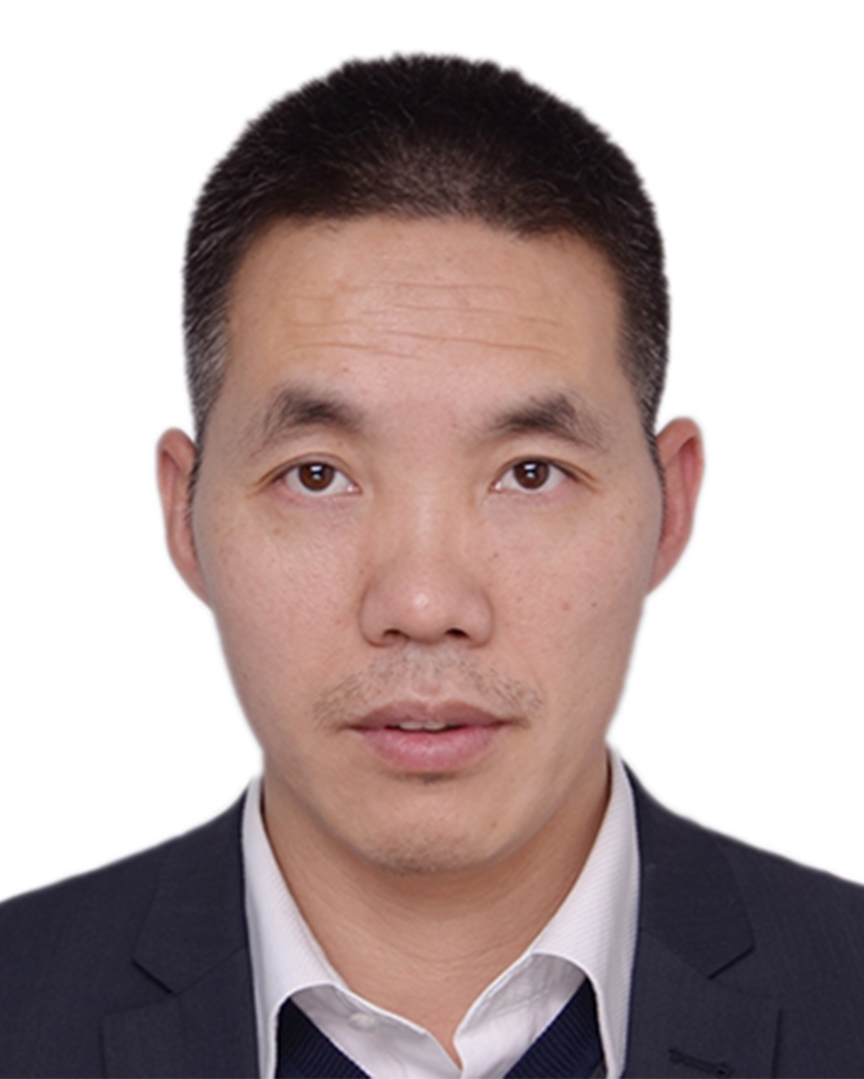}}]{Changzhi Wu} is a Chair Professor at the National Center for Applied Mathematics in Chongqing, China, and serves as a doctoral supervisor. He is currently the Managing Editor of the Journal of Industrial and Management Optimization. His research interests include big data modeling and analysis, optimization modeling and computational methods, and intelligent decision analysis.	
\end{biography}
\begin{biography}[{\includegraphics[width=1in,height=1.25in,clip,keepaspectratio]{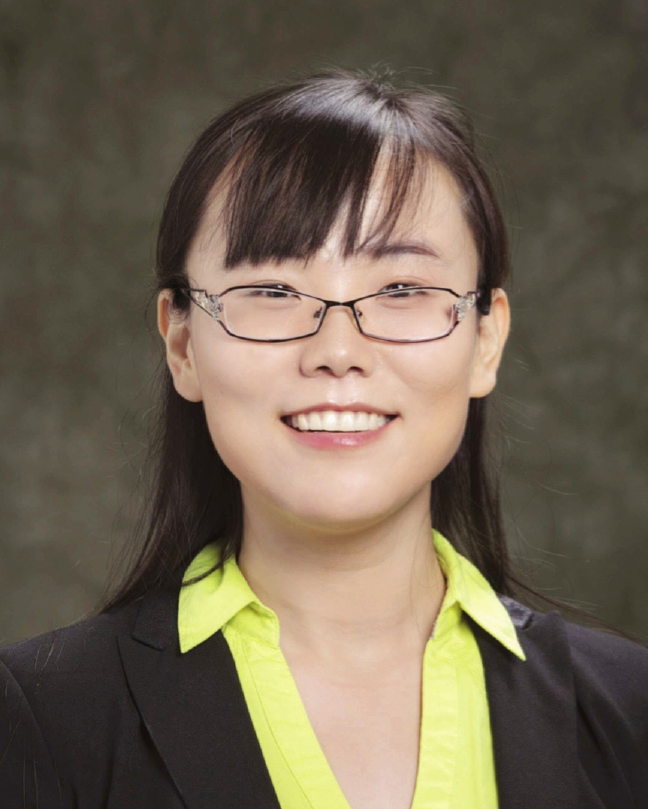}}]{Kaiyan Yu} (Member, IEEE) received the B.S. degree in intelligent science and technology from Nankai University, Tianjin, China, in 2010. She earned the Ph.D. degree in mechanical and aerospace engineering from Rutgers University, Piscataway, NJ, USA, in 2017. She joined the Department of Mechanical Engineering at Binghamton University, Binghamton, NY, USA, in 2018, where she is currently an Associate Professor. Her current research interests include autonomous robotic systems, motion planning and control, mechatronics, and automation science and engineering, with applications focused on Lab-on-a-chip technologies and nano/micro particle control and manipulation.
Dr. Yu is a member of the American Society of Mechanical Engineers (ASME). She is a recipient of the 2022 US NSF CAREER Award. She currently serves as an Associate Editor of the {\sc IEEE Transactions on Automation Science and Engineering}, {\em IEEE Robotics and Automation Letters}, {\em IFAC Mechatronics}, {\em Frontiers in Robotics and AI}, and the IEEE Robotics and Automation Society Conference Editorial Board and the ASME Dynamic Systems and Control Division Conference Editorial Board  (since 2018).		
\end{biography}


\vfill

 
\end{document}